\documentclass[journal,draftcls,onecolumn,12pt,twoside]{IEEEtranTCOM}
\usepackage{pgfplots}
\usepackage{epsfig}
\usepackage{amsthm,amssymb,graphicx,graphicx,multirow,amsmath,color}
\usepackage{caption}
\usetikzlibrary{plotmarks}

\def\E{\mathbb{E}}

\def\P{\mathbb{P}}

\def\ie{{\em i.e.}}

\def\R{\mathbb{R}}

\def\L{\mathcal{L}}

\def\v{\mathbf{v}}
\def\Q{\mathbf{Q}}

\def\d{\mathrm{d}}

\def\pzf{\mathtt{P}_{\mathtt{PZF}}}
\def\pmmse{\mathtt{P}_{\mathtt{MMSE}}}
\def\p{\mathtt{P}}

\def\I{\mathbf{I}}

\def\o{\hat{o}}

\def\l{\ell}

\newcommand{\hh}[1]{\ensuremath{\mathbf{h}_{#1}}}

\def\sinr{\mathtt{SINR}}

\def\sir{\mathtt{SIR}}

\newtheorem{lemma}{Lemma}{}
  
  \newtheorem{theorem}{Theorem}
  \newtheorem{prop}{Proposition}

\author{Sreejith~T.~Veetil, 
        Kiran~Kuchi  and~ Radha Krishna Ganti
\thanks{Sreejith T. V. and Kiran Kuchi are  with the Department
of Electrical Engineering, IIT Hyderabad, India. Radha Krishna Ganti is  with the Department
of Electrical Engineering, IIT Madras, India.}
}

\title{Performance of PZF and MMSE Receivers in Cellular Networks with Multi-User  Spatial Multiplexing}

\begin{document}
\maketitle  

\begin{abstract}
This paper characterizes the performance of cellular networks employing multiple antenna open-loop
 spatial multiplexing  techniques. We use a stochastic geometric framework to model distance depended inter cell interference. 
Using this framework, we analyze the coverage and rate using two linear
receivers, namely, partial zero-forcing (PZF) and minimum-mean-square-estimation (MMSE) receivers. Analytical expressions are obtained for coverage and rate distribution that are suitable for fast numerical computation. 

In the case of the PZF receiver,   we show that it is not optimal to utilize all the receive antenna for canceling interference. With $\alpha$ as the path loss exponent, $N_t$ transmit antenna, $N_r$ receive antenna, we show that it is optimal to use $N_t\left\lceil\left(1-\frac{2}{\alpha}\right) \left( \frac{N_r}{N_t} -\frac{1}{2}\right)\right\rceil$ receive antennas for interference cancellation and the remaining antennas for signal enhancement (array gain).  For both PZF and MMSE receivers, we observe that increasing the number of data streams  provides an improvement in the mean data rate with diminishing returns. Also transmitting $N_r$  streams is not optimal in terms of the mean sum rate.  
We observe that increasing the SM rate with a PZF receiver always degrades the cell edge data rate while the performance with MMSE receiver is nearly independent of the SM rate.
\end{abstract}
\section{Introduction}
 Multiple-input multiple-output  (MIMO) communication  is  an integral part of current cellular  standards. The  antennas can be used for increasing the number of data streams or improving the link reliability and the trade-off is well understood for a point-to-point link \cite{Telatar99capacityof, foschini1998limits,Tarokh98space-timecodes}.
Spatial multiplexing (SM)  is an important technique for boosting spectral efficiency of a point-to-point link with  multiple antenna, wherein  independent data streams are transmitted on different spatial dimensions.  The capacity improvement with SM in an isolated  link  in the presence of  additive Gaussian noise has been extensively studied \cite{1261332,1397926, gesbert}.

 However cellular systems are multi-user systems and co-channel interference is a  major impediment to the network performance.  Indeed, it has been argued with the help of simulations in \cite{CatreuxDriessenGreenstein,andrews2007overcoming}  that SM  is not very effective in  a multi-cell environment due to interference.   In a multi-user setup, in addition to providing diversity and multiplexing, the antennas can also  be used to serve different users and  {\em suppress} interference,  thereby adding  new dimensions for system design.

Existing cellular networks employ both closed-loop and open-loop SM methods.
Closed-loop SM requires channel state information (CSI) at the transmitter, and is suitable for users with
slowly varying channels (low Doppler case).  On the other hand open-loop SM is used  for
channels with high Doppler or in cases were there is inadequate feedback to support closed-loop SM.  Open-loop SM is also used for
increasing the performance of control channels where CSI feedback is not
available. Open-loop SM can be implemented in two ways. In single user SM, a
base station (BS) can allocate all the available data streams to a single user thereby increase the user's rate. Alternatively, the BS can serve multiple users at the
same time by allocating one stream per user. The latter approach is termed as
open-loop multi-user SM.

We consider the case where each BS has $N_t$ antenna  and  multiplexes
$N_t$ streams with one stream per user. The receiver  has $N_r$ receiver antennas.
In this case, with a linear receiver,  $N_t$ degrees-of-freedom (DOF) (among available $N_r$ DOF) can be used for suppressing
self-interference caused by SM while the remaining $N_r-N_t$ DOF can be used for
suppression of other cell interference or to obtain receiver array gain.
  For a typical cell edge user, a reduction of the SM rate at the transmitter might
result in enough residual DOF (after suppression of self-interference) to cancel the other cell interference and generally results in an
increased throughput. When the number of data streams transmitted from a BS is less than the number of antennas available, techniques like  cyclic-delay-diversity or open-loop dumb-beamforming \cite{Tse:2005:FWC} can be used  the remaining antennas. 

In this paper, we consider distance dependent inter-cell interference and investigate how multiple antenna can be used in the down link of an open-loop cellular system.    A general design goal is to maximize both mean and  cell edge data rates.  We analyze the various trade-offs between SM rate and the achievable mean/cell-edge data rates using linear receivers.


\section{Related Work}
Recent studies \cite{CatreuxDriessenGreenstein,andrews2007overcoming} show that spatial multiplexing MIMO systems, whose main benefit
is the supposed potential upswing in spectral efficiency, lose much of their
effectiveness in a multi-cell
environment with high interference.  Several approaches to handling interference in multi-cell MIMO systems are discussed in \cite{andrews2007overcoming}.  Blum in \cite{blum} investigated the capacity of an open-loop multi user MIMO system with
interference and have shown that the optimum power allocation across antenna depends on the interference power.  When the interference is high, it is optimal to allocate the entire power to one transmit antenna (single-stream) rather than spreading the power equally across antenna. 

There has been considerable work in ad hoc networks contrasting single stream transmission with multi-stream transmission using tools from stochastic geometry. It has been shown in   \cite{jindal,jindal-multi} that the network-wide throughput can be increased
linearly with the number of receive antennas, even if only a single transmit antenna is used by each node, and each node sends/receives only a single data stream. Interestingly, no channel state information (CSI) is required at the transmitter to achieve this gain.

Using  $(1-2/\alpha)$, where $\alpha$ is the path loss exponent, fraction of the receive degrees of
freedom for interference cancellation and the remaining degrees of freedom for array gain, allows for a  linear  scaling of the achievable rate with the  number of receiving antennas \cite{jindal-multi}. It is interesting to see that canceling
merely one interferer by each node may increase the transmission capacity even
when the CSI is imperfect \cite{HuangAndrews_SIC}. 
Importance of interference cancellation in
ad-hoc networks is also discussed in \cite{vaze,Andrews_Inte_Canc} and
\cite{DebbahHeath_intmit}.  However, most of these results are obtained by deriving bounds on the signal-to-interference ratio ($\sir$) distribution.

In \cite{louie2010spatial,Jindal_ICC} the exact distribution of $\sir$ with SM and minimum-mean-square-estimation (MMSE) receiver has been  obtained in an ad hoc network when the interferers are distributed as a spatial Poisson point process. The $\sir$ follows a quadratic form, and results from  \cite{gao1998theoretical, khatri} are used to obtain the distribution. Again, it was shown that single-stream transmission is preferable over
multi-stream transmission. In \cite{mckay}, distribution of $\sir$  for multiple antenna system with various receivers and  transmission schemes are obtained  for a Poisson interference field. In \cite{DebbahHeath_intmit} scaling laws for the
transmission capacity with zero-forcing beamforming were obtained,  and it was shown that for a large number of antennas, the maximum
density of concurrently transmitting nodes scales linearly with the number of
antennas at the transmitter, for a given outage constraint. In \cite{6655532}, the distribution of $\sir$ in a zero-forcing receiver  with co-channel interference is obtained.

In ad hoc networks, an interferer can be arbitrarily close (much closer than the intended transmitter) to the receiver
in consideration. This results in interference that is heavy-tailed. On the other hand, in a cellular network the user usually
connects to the closest BS and hence the distance to the nearest interferer is greater than the distance to the serving BS. This
leads to a more tamed interference distribution compared to an ad hoc networks.  

\subsection{Main Contributions}
In this paper we focus on  linear receivers, namely the 
 the  partial zero-forcing receiver and the  MMSE receiver. The MMSE receiver optimally balances signal
boosting and interference cancellation and maximizes the SINR. The 
sub-optimal partial zero-forcing receiver uses a specified number of
degrees of freedom for signal boosting and the remainder for interference cancellation. 
\begin{itemize}
\item We provide the distribution of $\sinr$ with a partial zero-forcing
receiver. This analysis  also includes the inter-cell interference which is
usually neglected.
  The resultant expression can be computed by evaluating a
single integral. 

\item  With one stream per-user, we obtain the optimal configuration
of receive antennas. In particular  we show
that it is optimal to use 
$N_t\left\lceil\left(1-\frac{2}{\alpha}\right) \left( \frac{N_r}{N_t} -\frac{1}{2}\right)\right\rceil$ receive antennas for interference
cancellation and the remaining antennas  for signal enhancement (array-gain).
\item  We compute the cumulative distribution function of the $\sinr$ with a linear
MMSE receiver. In the interference-limited case,  the distribution can be
computed without any integration. 
\item The sum rate expressions are provided for both PZF and MMSE receivers. Numerical evaluation of these results show that average sum rate increases with the number of data streams with diminishing returns.  The mean sum rate reaches a maximum value
for a certain optimum  number of data stream that is generally less than the number of receive antenna $N_r$. On the other hand, increasing the number of data streams always degrades the cell edge data rate for PZF receiver and MMSE receivers. However, the impact is less severe with a MMSE receiver. 
\end{itemize}
\subsection{Organization of the paper}
In Section \ref{sec:sysmodel}, the system model, particularly the BS location
model, is described in detail. In Section \ref{sec:pzf}, the $\sinr$ distribution with  a partial-zero-forcing receiver is derived. In Section \ref{sec:mmse}, the $\sinr$ distribution is obtained for a
linear MMSE receiver and in Section \ref{sec:rate}, the average
ergodic rate is analyzed with  both PZF and MMSE receivers.  The paper is concluded in Section \ref{conclusions}.

\section{System Model}
\label{sec:sysmodel}
We  now provide  a mathematical model of the cellular system that will be used in the subsequent  analysis. We begin  with the spatial distribution of the base stations. 

\noindent{\em Network Model}: The locations of  the  base stations (BSs) are modeled by a spatial Poisson point process (PPP)  \cite{stoyan} $\Phi \subset \R^2$ of density $\lambda$.  The  PPP model for BS spatial location provides a good approximation for irregular BS deployments.  The merits and demerits of this model for BS locations have been extensively discussed in \cite{ganti_coverage}. 

We assume the nearest BS connectivity model, \ie, a user  connects to the nearest BS. This nearest BS connectivity model  results in a Voronoi tessellation of the plane with respect to the BS locations.  See Figure \ref{fig:vor}. Hence  the service area of a BS is the Voronoi cell associated with it. 

We assume that  each BS  is equipped with $N_t$ antenna (active transmitting antennas) and a user (UE) is equipped with $N_r$ antenna.  In this paper we focus on downlink and hence the $N_t$ antenna at the BSs are used for transmission and the $N_r$ antenna at the UE are used for reception.  We assume that all the BSs transmit with equal power which for convenience we set to unity. Hence each transmit antenna uses a power of $1/N_t$.   

\noindent{\em Channel  and path loss model}:  We assume independent Rayleigh fading  with unit mean between any pair of antenna.  We focus on the downlink performance and hence  without loss of generality, we consider and analyse the  performance of  a typical mobile user located at the origin.   The $N_r \times 1$ fading vector between the $q$-th antenna of the BS $x \in \Phi$  and the typical mobile at the origin is denoted by $ \hh{x, q}$. We assume $\hh{x, q} \sim \mathcal{CN}(\mathbf{0}_{N_r \times 1},\mathbf{I}_{N_r})$, \ie, a circularly-symmetric complex Gaussian random vector.   The standard  path loss model $\l(x)=\|x\|^{-\alpha}$,  with path loss exponent $\alpha>2$ is assumed. Specifically, the link between $q$-th transmit antenna of the BS at $x$ and the $N_r$ receiver antennas of the user at origin  is $\sqrt{\|x\|^{-\alpha}}\hh{x,q}$.

\noindent{\em Received signal and interference}:
 We consider the case where each BS uses its $N_t$ antennas to serve  $N_t$  {\em independent}  data  streams to $N_t$ users in its cell\footnote{We make the assumption that every cell has at least $N_t$ users. This is true with high probability when there are large number of users which is normally the case. }.   Let $\o \in \Phi$ denote the BS that is closest to the mobile user at the origin.  We assume that the UE at the origin is interested in decoding the $k$-th stream transmitted by its associated  BS $\hat{o}$. Focusing on the $k$-{th} stream transmitted by $\o$, the received $N_r \times 1$ signal vector  at the typical mobile user is  
 \begin{align}
\mathbf{y}=&   \frac{a_{\o,k}}{\sqrt{r^{\alpha}}}\hh{\o,k} +\frac{1}{\sqrt{r^\alpha}}\sum_{q=1,q\ne k}^{N_t}\hh{\o,q} a_{\o,q}+\I(\Phi) +\mathbf{w},
\label{eq:main1}
 \end{align}
where  $\I(\Phi)= \sum_{x\in \Phi\setminus \o} \frac{1}{\sqrt{\|x\|^\alpha}} \sum_{q=1}^{N_t} \hh{x,q} a_{x,q}$,
denotes the intercell interference from other BSs.   The  symbol transmitted from the the $q$-th  antenna of the base station  $x \in \Phi$  is denoted by $a_{x,q}$ and  $\E[|a_{x,q}|^2]=1/N_t$.  The additive white Gaussian noise is given by $\textbf{w} \sim \mathcal{CN} (\mathbf{0}_{N_r \times 1},\sigma^2 \mathbf{I}_{N_r}) $.  The distance between the typical mobile user at the origin and its associated (closest) BS is denoted by $r = \|\o\|$.   Observe that $r$ is a random variable since the BS locations  are random.
We now present  few auxiliary results on the distribution of some spatial random variables that will be used later in the paper. 
\begin{figure}
\begin{center}
\includegraphics[width=7cm]{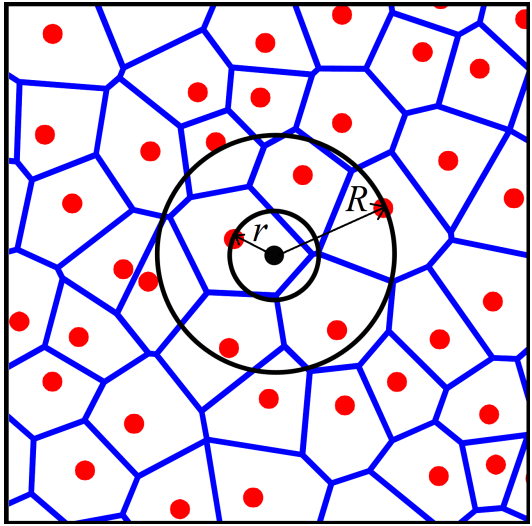}
\end{center}
\caption{Illustration of the BS locations modelled by a PPP and their corresponding cells. As an example,  the distance to the fifth nearest BS  from the typical user at the origin is denoted by $R$ and the distance to the serving BS is denoted by $r$.}
\label{fig:vor}
\end{figure}

\noindent{\em Distance to the serving BS and $(m-1)$-th interfering  BS}:
 We now  obtain the  joint distribution of the distance of the  origin to the nearest BS and the distance to the  $m-1$ interfering BS. 
Recall that $r$ denotes the distance to the serving (nearest) BS. The   PDF of the distance to the nearest neighbor is \cite{stoyan}
\begin{equation}
f_r(r)= e^{-\lambda \pi r^2}2\pi\lambda r.
\label{eqn:fr}
\end{equation}
 We now compute the distance to the $(m-1)$-th closest interfering  BS conditioned on the distance to the nearest BS $r$. 
Let $R$ denote the distance to the $(m-1)$-th BS. See Figure \ref{fig:vor}. Hence the event $R \leq R_0$ equals the event that there are at least $m-1$ base stations in the region between two concentric circles of radius $r$ and $R$ centred at origin. Hence
 \begin{align}\label{eqn:fRr}
F_{R|r}(R_0\ |\ r_0)=&\:\mathbb{P}\left[R\le R_0|r=r_0\right]=\sum_{k=m-1}^{\infty}e^{-\pi\lambda(R_0^2-r_0^2)}\displaystyle\frac{[\lambda\pi(R_0^2-r_0^2)]^k}{k!}, \ R_0>r_0.
\end{align} 
 Hence the conditional PDF is
\begin{align}
f_{R|r}(R|r)
=&\frac{2\pi\lambda R}{(m-2)!} e^{-\pi\lambda(R^2-r^2)}
\left(\pi\lambda(R^2-r^2)\right)^{m-2} , \ R>r.
\label{eqn:fRr2}
\end{align}
Let $\beta = R/r$ denote the ratio of the distance of the $m-1$ th closest interfering  BS of the typical user to the distance of its closest BS.  Using \eqref{eqn:fRr2} and \eqref{eqn:fr}, it can be easily shown that the PDF of the random variable $\beta$ is
\begin{align}
g_\beta(\beta) =2(m-1)\beta^{1-2m}(\beta^2-1)^{m-2}, \ \beta>1.
\label{eq:gb} 
\end{align}
Observe that the ration $\beta$ does not depend on the density of the PPP. The average value of $\beta$ is given by $\E[\beta]=\frac{\sqrt{\pi} \Gamma(m)}{\Gamma(m-1/2)}\approx \sqrt{(m-1)\pi}$,
and $\E[\beta^\nu] =\infty, \ \nu\geq 2$ irrespective of $m \geq2$.
 We begin with the analysis of a partial  zero forcing receiver. 

\section{Partial zero forcing (PZF) receiver}
\label{sec:pzf}
 In this section, we will analyze the distribution of the post-processing $\sinr$ with a PZF receiver in a cellular setting.  We assume that the user  has perfect knowledge of the interfering node channels that it wishes to cancel. 
\subsection{Coverage probability }
 Each user  has $N_r$ antenna, which can be represented as $N_r = m N_t+\delta$,  $m\geq 1$   receive antenna.  
%
The  receive filter $\v$ for the typical user at the origin is chosen orthogonal to the channel vectors of the interferers  and the streams that need to be canceled. 
Without loss of generality, we assume that the typical UE at the origin is  interested in the  $k$-th stream \eqref{eq:main1}. 
The receive filter  $\v$ is chosen as a unit norm vector orthogonal to the following vectors:
\begin{align*} 
\hh{\o,q}&:q=1, 2, .. ,k-1, k+1,.., N_t,\\
\hh{x,q}&:x\in \{x_1,x_2,...,x_{m-1}\}, \ q=1, 2, .., N_t,
\end{align*}
where $ \{x_1,x_2,...,x_{m-1}\}$ are the $(m-1)$ BSs closest to the typical UE  in consideration excluding $\o$. 
The dimension of the span of the above mentioned vectors is $N_t-1+(m-1)N_t = m
N_t-1$  with high probability.  Amongst the filters orthogonal to those vectors,
we are interested in the one that maximizes the signal power $|\v^\dag \hh{\o,k}
|^2$ . This corresponds to choosing $\v$ in the direction of the projection of
vector $\hh{\o,k}$ on the nullspace of the interfering channel vectors. The
dimension of the corresponding nullspace is $N_r-mN_t + 1$. If the columns of
an $(N_r-mN_t + 1) \times N_r$ matrix $\Q$ form an orthonormal basis for
this nullspace, then the receive filter $\v$ is chosen as:
\[
 \v=\Q\frac{\Q^\dag \hh{\o,k}}{\|\Q^\dag \hh{\o,k}\|}.
\]
So if $N_r = m N_t+\delta$, the remaining
$\delta+1$ degrees of freedom can be used to boost the signal power.  Hence at
the receiver, 
 \begin{align*}
\v^\dag\mathbf{y}_{k}=&  \frac{a_{\o,k}}{\sqrt{r^{\alpha}}} \v^\dag\hh{\o,k} +\sum_{q=1,q\ne k}^{N_t} \frac{a_{\o,q}}{\sqrt{r^\alpha}}\v^\dag\hh{\o,q}+\v^\dag\I(\Phi)+\v^\dag\mathbf{w}.
 \end{align*}
Since $\v^\dag$ is designed to null the closest $m-1$ interferers,  $\v^\dag\I(\Phi)=\v^\dag\I(\hat{\Phi}) $ where $\hat{\Phi} = \Phi\setminus \{x_1,...,x_{m-1}\}$. So we have 
$\tilde{y}_{k}=  \frac{a_{\o,k}}{\sqrt{r^{\alpha}}} \v^\dag\hh{\o,k} +\v^\dag\I(\hat{\Phi})+ \v^\dag\mathbf{w}$.
Let  $S \overset{\triangle}{=} |\v^\dag\hh{\o,k}|^2$ and $H_{x,q}\overset{\triangle}{=} |\v^\dag\hh{x,q}|^2$. 
The post processing zero-forcing signal-to-interference-noise ratio ($\sinr$)  \cite{jindal} of the $k$-th stream  is 
\begin{align}
\sinr=&\frac{Sr^{-\alpha}} {N_t \sigma^2 + \underbrace{\sum_{x\in \hat\Phi}\|x\|^{-\alpha} \sum_{q=1}^{N_t} H_{x,q}}_{\hat\I(\hat\Phi)}}.
\label{eqn:SINR}
\end{align}
 Also  $S \sim \chi^2_{2(N_r-m N_t+1)}$, \ie, a chi-squared random variable with $2(N_r-m N_t+1)$ degrees of freedom and $H_{x,q}$ are  i.i.d. exponential random variables.  When $N_r=mN_t$ the receiver can only cancel interference from $(m-1)$ nearest BSs and in this case, $S$   is an exponential random variable.
 
  A mobile user is said to be in coverage if the received  $\sinr$ (after pre-processing) is greater than the threshold $z$, needed to establish the connection. The probability of coverage is  defined as 
\begin{equation}\label{eqn:CovDef}
\pzf(z)\overset{\triangle}{=}\P[\sinr>z].
\end{equation} 
Observe that the coverage is essentially the complementary cumulative distribution function (CCDF) of the $\sinr$. Since  $\pzf(z)$, quantifies the entire distribution, it can be used to compute other metrics of interest like average ergodic rate. 
We first provide the main result which deals with the coverage probability with noise. We begin with the  evaluation of  the Laplace transform of interference conditioned on  the distances $R$ and $r$. 

\begin{lemma}
The Laplace transform of the residual interference in PZF conditional on $R$ and $r$ is given by 
\[\L_{\I}(s)=\exp\left(-\lambda \pi   R^2 \left(\, _2F_1\left(N_t,-\frac{2}{\alpha
   },\frac{\alpha -2}{\alpha },-R^{-\alpha } s\right)-1\right)\right),\]
where ${}_2F_1(a,b,c,z)$ is the standard hypergeometric function\footnote{${}_2F_1(a,b,c,z)=\frac{\Gamma(c)}{\Gamma(b)\Gamma(c-b)}\int_0^1\frac{t^{b-1}(1-t)^{c-b-1}}{(1-tz)^a}\d t$, and $\Gamma(x)$ is the standard Gamma function. }. 
\label{lem:Lap1}
\end{lemma}
\begin{proof}
The PZF receiver is designed such that it can cancel interference from nearest $(m-1)$ BSs. Conditioned on the distance to $(m-1)$-{th} BS  $R$, 
\begin{align*} 
\L_{\I}(s)=&\mathbb{E}\left[e^{-s\hat\I(\hat\Phi)} \right]  =\mathbb{E} \exp\left(-s\sum_{x \in \hat\Phi } \|x\|^{-\alpha} \sum_{q=1}^{N_t}H_{x,q}\right). \nonumber
\end{align*}
Since $H_{x,q}$ are i.i.d exponential, their sum  $\sum_{q=1}^{N_t}H_{x,q}$ is  Gamma distributed.  Using the Laplace transform of the Gamma distribution,
\begin{align}
\L_{\I}(s)=&\E\prod_{x\in\hat\Phi}\E \exp\left(-s\|x\|^{-\alpha}\sum_{q=1}^{N_t}H_{x,q}\right) =\E\prod_{x\in \hat\Phi }\frac{1}{(1+s\|x\|^{-\alpha})^{N_t}}\nonumber,\\
\stackrel{(a)}{=}&\exp\left(-\lambda 2\pi\int_R^\infty\left( 1- \frac{1}{(1+sx^{-\alpha})^{N_t}}\right) x\d x\right),\nonumber\\
=&\exp\left(-\lambda \pi   R^2 \left(\, _2F_1\left(N_t,-\frac{2}{\alpha
   },\frac{\alpha -2}{\alpha },-R^{-\alpha } s\right)-1\right)\right).
\label{eq:laplace}
\end{align}
where $(a)$ follows from the probability generating functional (PGFL)  of the PPP \cite{stoyan}. 
\end{proof}
The Laplace transform in Lemma \ref{lem:Lap1} is used next to compute the coverage probability. 
\begin{theorem}
The probability of coverage  with PZF receiver having $N_r=m N_t+\delta$ antennas is given by
\[\pzf(z)= \int_{0}^\infty \int_{r}^\infty \sum_{k=0}^{\delta}  \frac{(-s)^k}{k!} \frac{\d^k}{\d s^k} \L_{\I}(s)e^{-sN_t\sigma^2}\Big|_{s=z r^\alpha}f_{R|r}(R|r) f_r(r) \d R   \d r,\]
where $\L_{\I}(s)$ is given in Lemma \ref{lem:Lap1} and $f_{R|r}(R|r)$  in \eqref{eqn:fRr2}.
\label{the:cov_gen}
\end{theorem}
 \begin{proof}
Conditioned on the random variables $R$ and $r$,  we have
\begin{align} \p_c(z,\alpha\ |\ r)=& \mathbb{E}_{\I}[\mathbb{P}(S>z\, r^\alpha (\I +N_t\sigma^2) )]\nonumber,\\
\overset{(a)}{=}&\sum_{k=0}^{\delta} \frac{({z\, r^\alpha})^k }{k!} \mathbb{E}_{\I}[(\I +N_t\sigma^2)^k e^{-z r^\alpha (\I +N_t\sigma^2)}]\nonumber,\\
\overset{(b)}{=}&\sum_{k=0}^{\delta} \left( \frac{(-s)^k}{k!} \frac{\d^k}{\d s^k} \L_{(\I +N_t\sigma^2)}(s)\right)_{s=z r^\alpha}.
\label{eq:condprob}
\end{align}
$(a)$ follows from the CCDF of   $\chi_{2(N_r-m N_t+1)}^2$ and $(b)$ by the differentiation property of the Laplace transform.
$\L_{(\I +N_t\sigma^2)}(s)$ is the  Laplace transform of interference  and noise  and equals
$\L_{\I}(s)\exp(-sN_t\sigma^2)$ where $\L_{\I}(s)$ is given in Lemma \ref{lem:Lap1}. The result follows by averaging over $R$ and $r$.
\end{proof}
\subsection{Interference limited networks, $\sigma^2=0$}
\label{setpart}
We now specialize the coverage  expression in  Theorem \ref{the:cov_gen} when $\sigma^2=0$, \ie,  an interference-limited network. In Theorem \ref{the:cov_gen},  the coverage probability expression requires evaluating the $k$-th derivative of a composite function. The derivatives of a composite function can  be written in a succinct form by using a set partition version of Fa\`{a} di Bruno's formula. We now introduce some  notation that will be used in the next theorem.

 A partition of a set $S$ is a collection of disjoint subsets of $S$ whose union is $S$. The collection of all the set  partitions of the integer set $[1, 2, 3, \hdots k]$ is denoted by $\mathcal{S}_k$ and its cardinality is called the $k$-th Bell number. For a partition $\upsilon\in \mathcal{S}_k$, let  $|\upsilon|$ denote  the  number of blocks in the partition and $|\upsilon|_j$ denote  the number of blocks with exactly $j$ elements. For example, when $k=3$, there are $5$  partitions,  \[\mathcal{S}_3=\Big\{\{1,2,3\},\big\{\{1\},\{2,3\}\big\},\big\{\{1,2\},\{3\}\big\},\big\{\{1,3\},\{2\}\big\},\big\{\{ 1\},\{2\},\{3\}\big\}\Big\}.\] For the partition $ \upsilon=\big\{\{ 1\},\{2\},\{3\}\big\}$, $|\upsilon|=3$, $|\upsilon|_1=3$ and $|\upsilon|_2=0$.  Also, define \begin{equation*}
\Lambda_{\varsigma, N_t}(z) \triangleq
 \,_2F_1\left(N_t+\varsigma,\varsigma-\frac{2}{\alpha
},\varsigma-\frac{2}{\alpha
   }+1,-z \right),
\end{equation*}
where $\,_2F_1(a,b,c,x)$ is the standard hypergeometric function. 
\begin{theorem}
When the network is interference limited, \ie, $\sigma^2=0$, the probability of coverage  with a PZF receiver having $N_r=m N_t+\delta$ antennas is 
 \begin{align}
\pzf(z)=& \sum_{k=0}^{\delta} \frac{z ^k}{k!} \sum_{\upsilon\in \mathcal{S}_k}(-1)^{ |\upsilon|} (m)_{|\upsilon|} \E_\beta\left[ {  \frac{\beta^{-\alpha k } }{\Lambda_{0, N_t}(\beta^{-\alpha} z)^{m}}  \prod_{j=1}^k \left(\frac{(N_t)_j (-\frac{2}{\alpha
   })_j}{(\frac{\alpha -2}{\alpha })_j}\frac{\Lambda_{j, N_t}(\beta^{-\alpha} z)}{\Lambda_{0, N_t}(\beta^{-\alpha} z)}\right)^{|\upsilon|_j}}   \right],
   \label{eq:456}
\end{align} where, $(x)_n=\frac{\Gamma(x+n)}{\Gamma(n)}$ is the Pochhammer function., 
The expectation is with respect to the variable $\beta =R/r$  whose PDF is given by $g(\beta)$ provided in \eqref{eq:gb}

\label{the:cov}
\end{theorem}

\begin{proof}
See Appendix A. 
\end{proof}
 In Table \ref{tab:cm2k}, the coverage probability expressions are provided for the  case of  $m= 2$. We now make a  few observations: 
\begin{itemize}
\item When $m=1$, \ie, only the self-interference from the other data streams is canceled,  $\beta =1$ almost surely and hence the expectation with respect to $\beta$ in Theorem \ref{the:cov} can be dropped.
\item  When $\delta =0$, \ie,  all the antenna are used to cancel interference, then \[\pzf(z)=\E_\beta [\Lambda_{0, N_t}(\beta^{-\alpha} z)^{-m}].\]  From the expression, it seems that the coverage probability increases exponentially with the number of canceled interferers. However, this is not the case as $\beta$ in the above expression is a function of $m$.   This can be seen in Figure \ref{fig:PzfVsm} where  we observe diminishing benefits with increasing $m$
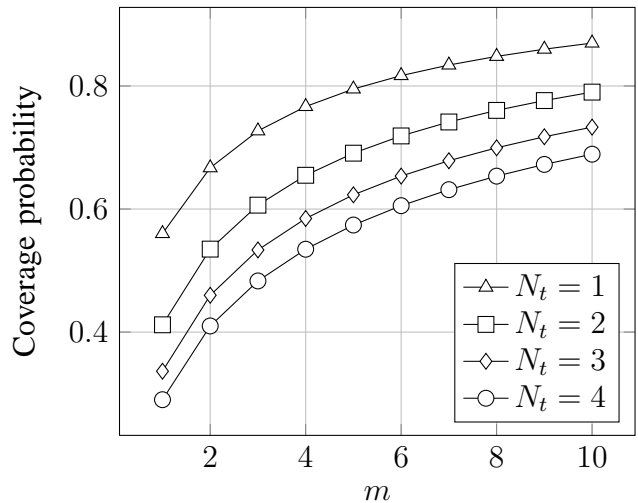
\begin{figure}
\centering 

\begin{tikzpicture}
\begin{axis}[scale=1,
grid = both,
legend pos =south east,
xlabel =$m$,
ylabel = Coverage probability 
]

 \addplot[ mark=triangle*, mark size=3,  mark options={fill=white} ]
 coordinates{
 (1.000000, 0.560099) (2.000000, 0.667024) (3.000000, 0.727027) (4.000000, 0.766596) (5.000000, 0.795110) (6.000000, 0.816845) (7.000000, 0.834070) (8.000000, 0.848117) (9.000000, 0.859827) (10.000000, 0.869761) 
  };\addlegendentry{$N_t=1$}
  \addplot[ mark=square*, mark size=3,  mark options={fill=white} ]
  coordinates{
  (1.000000, 0.411845) (2.000000, 0.534862) (3.000000, 0.606318) (4.000000, 0.654880) (5.000000, 0.690808) (6.000000, 0.718838) (7.000000, 0.741513) (8.000000, 0.760347) (9.000000, 0.776309) (10.000000, 0.790054) 
   };\addlegendentry{$N_t=2$}
   \addplot[ mark=diamond*, mark size=3,  mark options={fill=white} ]
   coordinates{
   (1.000000, 0.336403) (2.000000, 0.459810) (3.000000, 0.533492) (4.000000, 0.584612) (5.000000, 0.623088) (6.000000, 0.653554) (7.000000, 0.678524) (8.000000, 0.699507) (9.000000, 0.717478) (10.000000, 0.733102) 
    };\addlegendentry{$N_t=3$}
    \addplot[ mark= *, mark size=3,  mark options={fill=white} ]
    coordinates{
    (1.000000, 0.290088) (2.000000, 0.409949) (3.000000, 0.483135) (4.000000, 0.534701) (5.000000, 0.573999) (6.000000, 0.605446) (7.000000, 0.631460) (8.000000, 0.653503) (9.000000, 0.672523) (10.000000, 0.689172) 
     };\addlegendentry{$N_t=4$}
\end{axis}
\end{tikzpicture}
\caption{Coverage probability versus $m$ for various $N_t$ with $\alpha=4$ at $z=0$ dB.}
	\label{fig:PzfVsm}
\end{figure}

\item In Theorem \ref{the:cov}, the coverage probability is obtained by averaging over $\beta =R/r$. Hence $\pzf(z)$ corresponds to the coverage probability of typical user. Instead of averaging over $\beta$,  evaluating \eqref{eq:456} at a particular value of $\beta$ would indicate the coverage of a  user at a specified distance. For example, $\beta=1$  would  correspond to an edge user with $m=2$. 
\end{itemize}
 \begin{table} [h]
\centering
 \begin{tabular}{|c|l|}
\hline $N_t\times N_r$&Coverage probability $\pzf(z) $\\\hline 
 $1\times 2$ & $\int_{\beta=1}^\infty \frac{2}{\beta ^3 \, _2F_1\left(1,-\frac{2}{\alpha };1-\frac{2}{\alpha };-z \beta ^{-\alpha }\right){}^2}\d\beta$.\\ \hline 
$1\times 3$ & $ \int_{\beta=1}^\infty \frac{2}{\beta ^3 \, _2F_1\left(1,-\frac{2}{\alpha };1-\frac{2}{\alpha };-z \beta ^{-\alpha }\right){}^2}-\frac{4 z \beta ^{-\alpha
   -3} \Gamma \left(1-\frac{2}{\alpha }\right)^2 \, _2F_1\left(2,1-\frac{2}{\alpha };2-\frac{2}{\alpha };-z \beta ^{-\alpha
   }\right)}{\Gamma \left(2-\frac{2}{\alpha }\right) \Gamma \left(-\frac{2}{\alpha }\right) \, _2F_1\left(1,-\frac{2}{\alpha
   };1-\frac{2}{\alpha };-z \beta ^{-\alpha }\right){}^3} \d\beta$.\\ \hline 
$2\times 4$ & $\int_{\beta=1}^\infty \frac{2}{\beta ^3 \, _2F_1\left(2,-\frac{2}{\alpha };1-\frac{2}{\alpha };-z \beta ^{-\alpha }\right){}^2}\d\beta $.\\\hline
 $2\times 5$ & $\int_{\beta=1}^\infty \frac{2}{\beta ^3 \, _2F_1\left(2,-\frac{2}{\alpha };1-\frac{2}{\alpha };-z \beta ^{-\alpha }\right){}^2}-\frac{8 z \beta ^{-\alpha   -3} \Gamma \left(1-\frac{2}{\alpha }\right)^2 \, _2F_1\left(3,1-\frac{2}{\alpha };2-\frac{2}{\alpha };-z \beta ^{-\alpha   }\right)}{\Gamma \left(2-\frac{2}{\alpha }\right) \Gamma \left(-\frac{2}{\alpha }\right) \, _2F_1\left(2,-\frac{2}{\alpha   };1-\frac{2}{\alpha };-z \beta ^{-\alpha }\right){}^3} \d\beta$.\\ \hline
\end{tabular}
\caption{Coverage probability expressions for  $\sigma^2=0$, $m=2$  with different $N_t$.}
\label{tab:cm2k}
\end{table}
%

\subsection{Interference cancellation or signal enhancement?}
The antennas at the receiver can be used for either interference cancellation or enhancing the desired signal. In our formulation,  $mN_t-1$ antenna are used for interference cancellation while $\delta+1$ antenna are used for signal enhancement. For a given $N_r$, what is the optimal $(m, \delta)$ split to maximize the coverage probability? Since coverage probability is  a complicated expression of $(m,\delta)$, we will use the  average interference-to-signal ratio as the metric.  
We have
\[\E[\sinr^{-1}] = \E[N_t \sigma^2 r^\alpha S ^{-1} +S ^{-1} \I r^{\alpha} ],\]
which equals $N_t \sigma^2 \E[r^\alpha]\E[ S ^{-1}] +\E[S ^{-1}]\E[ \I r^{\alpha} ]]$.
If $\delta \neq 0$, then $\E[ S ^{-1}] = \frac{1}{2\delta}$. Since $r$ is Rayleigh distributed, 
$\E[r^\alpha] = (\pi\lambda)^{-\alpha/2}\Gamma(1+\alpha/2)$.
Also, 
$\E[ \I r^{\alpha} ]= N_t \sum_{k=m}^\infty\E[\beta^{-\alpha}_{(k)}]$,
where $\beta_{(k)}$ represent the ratio of the distance to the $k$-th nearest interfering station to the serving BS distance. Using \eqref{eq:gb}, we obtain
\[\E[ \I r^{\alpha} ]= N_t \sum_{k=m}^\infty{(k-1)\Gamma(k-1)\Gamma(1+\alpha/2)}{\Gamma(k+\alpha/2)^{-1}}.\] 
Hence 
\begin{equation}
\E[\sinr^{-1}]=\frac{N_t\Gamma(1+\alpha/2)}{2\delta}\left( \sigma^2(\pi\lambda)^{-\alpha/2}\Gamma(1+\alpha/2)+ \sum_{k=m}^\infty\frac{\Gamma(k)\Gamma(1+\alpha/2)}{\Gamma(k+\alpha/2)} \right).
\label{eq:act}
\end{equation}
It  follows from Kershaw's inequality that $\Gamma(k+\alpha/2)/\Gamma(k) \approx (k+\alpha/4-1/2)^{\alpha/2}$ (indeed an upper bound).  Substituting for $\Gamma(k+\alpha/2)/\Gamma(k)$ and replacing the summation by integration, we have the following approximation:
\begin{equation}
\E[\sinr^{-1}]\approx\frac{N_t\Gamma(1+\frac{\alpha}{2})}{2\delta}\left( \sigma^2(\pi\lambda)^{-\frac{\alpha}{2}}\Gamma(1+\frac{\alpha}{2})+ \frac{2\Gamma(1+\frac{\alpha}{2})(m+\frac{\alpha}{4}-1/2)^{1-\frac{\alpha}{2}}}{\alpha-2} \right).
\label{eq:567}
\end{equation}
Using this result we can obtain the optimal $m$ and is stated in the following proposition. 
\begin{figure}
\centering 
\begin{tikzpicture}
\begin{axis}[
grid = both,
legend pos =north east,
xlabel =$m$,
ylabel = {$\E[\sinr^{-1}]$}
]
  \addplot[ mark=triangle*, mark size=2,  mark options={fill=black} ]
coordinates{
 (2,0.122525)  (3,0.0924092)  (4,0.080033)  (5,0.0760396)  (6,0.0783828) (7,0.0886374) (8,0.115099)  (9,0.20242)
 };
  \addlegendentry{$1\times 10, \alpha=4$}
 
       \addplot[ mark=square*, mark size=2,  mark options={fill=black} ]
coordinates{
(1,0.201995)  (2, 0.144161) (3, 0.126774) (4, 0.122582) (5, 0.126842) (6, 0.140137) (7, 0.167962) (8, 0.229279) (9, 0.421227)
 };
  \addlegendentry{$1\times 10, \alpha=3$}
      \addplot[ mark=diamond*, mark size=2,  mark options={fill=black} ]
coordinates{
(1, 0.49505) (2, 0.326733) (3, 0.323432) (4, 0.480198)
 };
  \addlegendentry{$2\times 10, \alpha=4$}
  
        \addplot[ mark=*, mark size=2,  mark options={fill=black} ]
coordinates{
(1,0.454489) (2, 0.384428) (3, 0.443708) (4, 0.735492)
 };
  \addlegendentry{$2\times 10, \alpha=3$}
  
    \addplot[ mark=*, mark size=2,  mark options={fill=white} ]
coordinates{
(1,0.848656) (2,0.735149) (3,1.94059)
 };
  \addlegendentry{$3\times 10, \alpha=4$}
  
   \addplot[   mark =oplus,  mark size=5, only marks]
coordinates{
  (5,0.0760396)  (2,0.735149) (4, 0.122582) (3, 0.323432) (2, 0.384428)
 };
 \addlegendentry{Optimal}
\end{axis}
\end{tikzpicture}
\caption{Average $\E[\sinr^{-1}]$ from \eqref{eq:act} as a function of $m$ for $N_r=10$ with different $N_t$ and $\alpha$ for $\sigma^2=0$. The optimal (lowest) value in each case is marked. We observe that the optimal value is always  $N_t\left\lceil\left(1-\frac{2}{\alpha}\right) \left( \frac{N_r}{N_t} -\frac{1}{2}\right)\right\rceil$ as specified in Proposition \ref{prop:one}.}
	\label{fig:opt_m}
\end{figure}
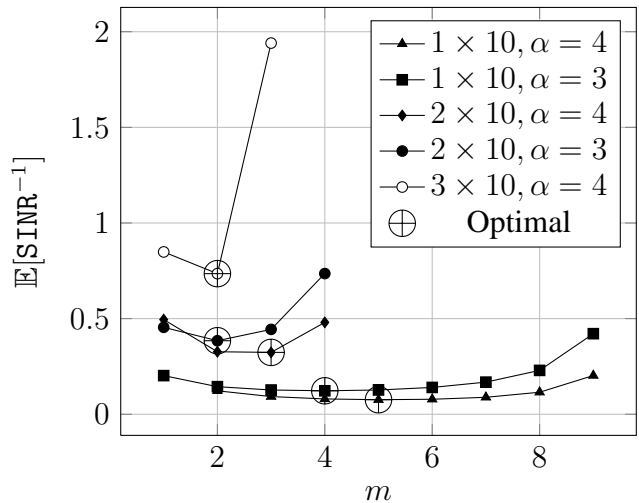
\begin{prop}
\label{prop:one}
The optimal $m^*(\sigma^2)$ is the smallest integer that is greater than the positive root of the equation 
\[2(m+\alpha/4-1/2)^{1-\alpha/2}N_t +\sigma^2(\pi\lambda)^{-\alpha/2}N_t(\alpha-2)-(N_r-mN_t)(m+\alpha/4-1/2)^{-\alpha/2}(\alpha-2)=0,\]
if such a root exists.  In particular when $\sigma^2=0$, \ie, when the system is interference limited, we have 
\begin{align}
m^*(0)= N_t\left\lceil\left(1-\frac{2}{\alpha}\right) \left( \frac{N_r}{N_t} -\frac{1}{2}\right)\right\rceil .
\label{eq:opt}
\end{align}
\end{prop}
\begin{proof}
To find the optimal $m$, we set  $\delta = N_r-mN_t$ in \eqref{eq:567}, differentiate and equate to zero. 
\end{proof}

This is in tune with the results in \cite{jindal}, where they show that it is optimal to use  $(1-2/\alpha)$ fraction of the antennas for interference cancellation.  In Figure \ref{fig:opt_m}, the average $\E[\sinr^{-1}]$ computed using \eqref{eq:act}  is plotted as a function of $m$ for various configurations. We observe that the optimal $m^*(0)$ coincides with the optimal $m$  as can be seen in the Figure \ref{fig:opt_m}.  This result indicates that it is optimal to utilize $m^*(0)$ fraction of receive antenna to cancel interfering nodes and utilize the remaining antenna to strengthen the desired signal. So as the path loss exponent $\alpha$ increases, more antenna should be used for interference cancellation rather than boosting the desired signal.  

\subsection{Numerical results for coverage and discussion}
\begin{figure}
\centering

\begin{tikzpicture}
\begin{semilogyaxis}[scale=1.8,
xmin=-2,
xmax=16,
ymin=0.05,
grid = both,
legend style={
 cells={anchor=west},
legend pos=south west ,
 },
xlabel =Threshold $z$ dB,
ylabel = Coverage probability 
]

\addplot[ mark=square*, mark size=3,  mark options={fill=white}]
coordinates{
(-5.000000, 0.931920) (-4.000000, 0.902189) (-3.000000, 0.864880) (-2.000000, 0.820385) (-1.000000, 0.769847) (0.000000, 0.714995) (1.000000, 0.657849) (2.000000, 0.600398) (3.000000, 0.544351) (4.000000, 0.490993) (5.000000, 0.441161) (6.000000, 0.395293) (7.000000, 0.353519) (8.000000, 0.315756) (9.000000, 0.281794) (10.000000, 0.251352) (11.000000, 0.224126) (12.000000, 0.199810) (13.000000, 0.178111) (14.000000, 0.158757) (15.000000, 0.141500) (16.000000, 0.126116) (17.000000, 0.112403) (18.000000, 0.100181) (19.000000, 0.089287) (20.000000, 0.079577) 
 };
 \addlegendentry{$2\times4$, $m^*(0)=1$} 
 \addplot[ mark=triangle*, mark size=3,  mark options={fill=white} ]
coordinates{
(-5.000000, 0.767218) (-4.000000, 0.715645) (-3.000000, 0.661176) (-2.000000, 0.605583) (-1.000000, 0.550560) (0.000000, 0.497532) (1.000000, 0.447542) (2.000000, 0.401230) (3.000000, 0.358886) (4.000000, 0.320534) (5.000000, 0.286016) (6.000000, 0.255077) (7.000000, 0.227416) (8.000000, 0.202720) (9.000000, 0.180690) (10.000000, 0.161047) (11.000000, 0.143536) (12.000000, 0.127928) (13.000000, 0.114017) (14.000000, 0.101618) (15.000000, 0.090567) (16.000000, 0.080718) (17.000000, 0.071940) (18.000000, 0.064117) (19.000000, 0.057144) (20.000000, 0.050930) 
 };
 \addlegendentry{$3\times4$, $m^*(0)=1$} 
 \addplot[  mark=*, mark size=3, mark options={fill=white} ]
coordinates{
(-5.000000, 0.494504) (-4.000000, 0.447898) (-3.000000, 0.403698) (-2.000000, 0.362482) (-1.000000, 0.324578) (0.000000, 0.290088) (1.000000, 0.258949) (2.000000, 0.230984) (3.000000, 0.205953) (4.000000, 0.183594) (5.000000, 0.163644) (6.000000, 0.145854) (7.000000, 0.129995) (8.000000, 0.115859) (9.000000, 0.103260) (10.000000, 0.092031) (11.000000, 0.082022) (12.000000, 0.073103) (13.000000, 0.065153) (14.000000, 0.058067) (15.000000, 0.051753) (16.000000, 0.046124) (17.000000, 0.041108) (18.000000, 0.036638) (19.000000, 0.032654) (20.000000, 0.029103) 
 };
  \addlegendentry{$4\times4$, $m^*(0)=1$}

\addplot[ mark=square*, mark size=3,  mark options={fill=black}]
coordinates{
 (-4.999997 , 0.994012)  (-4.000002 , 0.988951)  (-3.000002 , 0.980573)  (-2.000002 , 0.967526)  (-1.000001 , 0.948449)  (0.000000 , 0.922262)  (1.000016 , 0.888456)  (1.999991 , 0.847283)  (2.999995 , 0.799750)  (4.000006 , 0.747443)  (5.000003 , 0.692240)  (5.999998 , 0.636015)  (6.999998 , 0.580417)  (7.999998 , 0.526746)  (8.999999 , 0.475921)  (10.000000 , 0.428511)  (11.000016 , 0.384797)  (11.999991 , 0.344849)  (12.999995 , 0.308584)  (14.000006 , 0.275827)  (15.000003 , 0.246347)  (15.999998 , 0.219887)  (16.999998 , 0.196186)  (17.999998 , 0.174986)  (18.999999 , 0.156042)  (20.000000 , 0.139127)   };
\addlegendentry{$1\times 4, m=1$}
\addplot[ mark=triangle*, mark size=3,  mark options={fill=black} ]
coordinates{
(-5.000000, 0.989629) (-4.000000, 0.983265) (-3.000000, 0.973881) (-2.000000, 0.960635) (-1.000000, 0.942763) (0.000000, 0.919708) (1.000000, 0.891244) (2.000000, 0.857531) (3.000000, 0.819108) (4.000000, 0.776811) (5.000000, 0.731667) (6.000000, 0.684760) (7.000000, 0.637138) (8.000000, 0.589733) (9.000000, 0.543322) (10.000000, 0.498518) (11.000000, 0.455769) (12.000000, 0.415382) (13.000000, 0.377540) (14.000000, 0.342329) (15.000000, 0.309759) (16.000000, 0.279780) (17.000000, 0.252303) (18.000000, 0.227211) (19.000000, 0.204366) (20.000000, 0.183623) 
 };
 \addlegendentry{$1\times 4, m=2$}%
\addplot[mark options={fill=black} ]
coordinates{
(-5.000000, 0.974010) (-4.000000, 0.963556) (-3.000000, 0.949894) (-2.000000, 0.932540) (-1.000000, 0.911135) (0.000000, 0.885501) (1.000000, 0.855685) (2.000000, 0.821968) (3.000000, 0.784833) (4.000000, 0.744925) (5.000000, 0.702985) (6.000000, 0.659783) (7.000000, 0.616071) (8.000000, 0.572532) (9.000000, 0.529762) (10.000000, 0.488252) (11.000000, 0.448388) (12.000000, 0.410453) (13.000000, 0.374641) (14.000000, 0.341070) (15.000000, 0.309788) (16.000000, 0.280793) (17.000000, 0.254044) (18.000000, 0.229465) (19.000000, 0.206961) (20.000000, 0.186421) 
 }; \addlegendentry{$1\times 4, m=3$}
 
  \addplot[style=dashed]
coordinates{
 (-4.999997 , 0.898226)  (-4.000002 , 0.877731)  (-3.000002 , 0.854292)  (-2.000002 , 0.827880)  (-1.000001 , 0.798581)  (0.000000 , 0.766596)  (1.000016 , 0.732243)  (1.999991 , 0.695931)  (2.999995 , 0.658142)  (4.000006 , 0.619395)  (5.000003 , 0.580215)  (5.999998 , 0.541107)  (6.999998 , 0.502531)  (7.999998 , 0.464890)  (8.999999 , 0.428519)  (10.000000 , 0.393685)  (11.000016 , 0.360585)  (11.999991 , 0.329354)  (12.999995 , 0.300072)  (14.000006 , 0.272772)  (15.000003 , 0.247446)  (15.999998 , 0.224056)  (16.999998 , 0.202538)  (17.999998 , 0.182813)  (18.999999 , 0.164787)  (20.000000 , 0.148359)  };
 \addlegendentry{$1\times 4, m=4$}

\addplot[
   mark =oplus,  mark size=5, only marks, mark options={fill=blue} ]
coordinates{
(-5.000000, 0.975450) (-2.000000, 0.934800) (1.000000, 0.859160) (4.000000, 0.749850) (7.000000, 0.620790) (10.000000, 0.494070) (13.000000, 0.379270) (16.000000, 0.283880) (19.000000, 0.209390) 
 };
 \addlegendentry{$1\times4$, $m=3$,  Sim}


\end{semilogyaxis}
\end{tikzpicture}
\caption{Coverage probability versus $z$ for $N_r=4$ and different $N_t$ with optimal choice of $m$. The path loss exponent $\alpha=4$ and $\sigma^2=0$.  Coverage probability versus $z$ for  different choices of $m$ and $\delta$. The path loss exponent $\alpha=4$ and $\sigma^2=0$. The Monte Carlo results are also plotted and  marked with $\oplus$. For $N_t \times N_r = 1\times 4$ and $\alpha =4$, from \eqref{eq:opt},  $m^*(0) =2$.}
	\label{fig:caseNr4}
\end{figure}
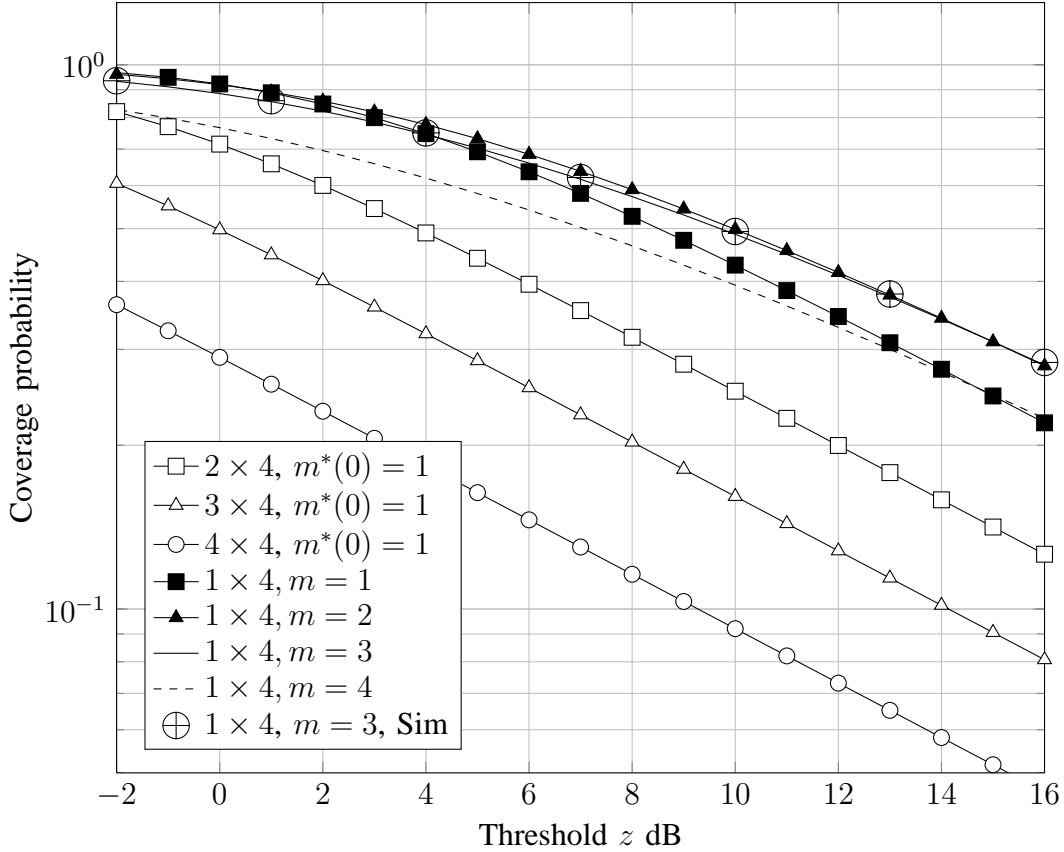
In Figure \ref{fig:caseNr4}, the coverage probability is plotted for $1\times 4$ configuration for $\alpha=4$ and different choices of $m$ and $\delta$ using Theorem \ref{the:cov}. In the same Figure, the coverage  results for different configurations obtained by Monte Carlo simulation are marked by $\oplus$. We first observe that the coverage results obtained with Monte Carlo simulations match with the analytical results. 

We observe that utilizing all the antennas for interference cancellation is not optimal. In fact, from   Figure  \ref{fig:caseNr4}, we observe that utilizing all the antenna for interference cancellation leads to the lowest coverage probability, particularly for  medium $\sir$ thresholds. 

If a very high $\sir$ is required, we observe from Figure \ref{fig:caseNr4} that  it is better to use all the antenna for cancellation. This can be observed by looking at the crossover points of the different curves. So an interior user should use his antenna to cancel interferers and obtain a higher $\sir$.  
For most users, canceling the strongest interference improves the coverage significantly over the  ZF receiver which  utilizes all of its receiver DoF to cancel interference.  So in a practical system, obtaining the channel of the nearest interferer is sufficient to have good coverage.  
It can also be seen that canceling nearest three BSs is giving almost same coverage compared to canceling nearest two, the reason is that the  interference from the third may not be strong enough. So the better strategy can be canceling nearest BSs and using the remaining DoF for array gain.  We also observe that canceling one interferer, \ie,  $m=2$ has  the highest coverage probability and corresponds to $m^*(0)$ for $N_t \times N_r =1 \times 4$.

Now coming to the performance of the PZF receivers with multi-stream
transmission, we can see that the coverage probability reduces  with 
increasing SM rate for a fixed number of antennas at the receiver. In Figure
\ref{fig:caseNr4}, it is easy to see that the coverage is
heavily reduced  with increasing the number of streams $N_t$.  This is because, 
increasing the number of streams while  keeping $N_r$ constant  will increase the
interference and the antenna available to cancel external interference is also
reduced. For the edge user this effect is dominant and we can get more insight
into this when we study the rate parameters. 

\section{Linear MMSE Receiver }
\label{sec:mmse}
In this Section, we analyze the performance of a linear MMSE receiver with inter-cell interference. We consider the case where each BS uses its antennas to serve  {\em independent}  data  streams to the users connected to it. 
Each user decodes its assigned stream using a linear MMSE receiver  treating other streams as interference.  Focusing on the user at the origin, interested in the $k$-th stream, the linear MMSE filter is given by $\v_{\o,k}^\dag= \hh{\o,k}^\dag\mathbf{R}_{\o,k}^{-1}$, where
\begin{align*}
 \mathbf{R}_{\o,k}=&\frac{1}{N_t\,r^{\alpha}}\sum_{q=1,q\ne k}^{N_t}\hh{\o,k}\,\hh{\o,k}^\dag  +\frac{1}{N_t}\sum_{x\in \Phi\setminus \o} \frac{1}{\sqrt{\|x\|^\alpha}}\mathbf{H}\mathbf{H}^\dag+\sigma^2\mathbf{I}_{N_r},
 \end{align*} is the interference plus noise covariance matrix.
The post processing $\sinr$ at the receiver is given by
\[ \sinr=\frac{1}{N_t \,r^{\alpha}}\hh{\o,k}^\dag\mathbf{R}_{\o,k}^{-1}\hh{\o,k}.\]
It is assumed that each receiving node has the knowledge of corresponding transmitting  channel $\mathbf{H}_0$ and $\mathbf{R}_{\o,k}$.


The result in \cite{gao1998theoretical, khatri} can be used to express the $\sinr$ distribution in terms of the channel gains. We then use   the probability generating functional of the  PPP to average the channel gains to obtain  the coverage with a MMSE receiver.    In \cite{Jindal_ICC},  the exact distribution of $\sir$ with SM and MMSE receiver has been  obtained in an ad hoc network when the interferers are distributed as a spatial Poisson point process. However, the results are obtained by starting with  a finite network and then obtaining the final distribution by a limiting argument.  The proof in this paper  uses  the probability generating functional, and  is easier  to extend  to other spatial distribution of nodes.  Also, as mentioned earlier, unlike an ad hoc network, where the distance to the intended transmitter is fixed, in a cellular network the distance $r$ is random  making the network scale invariant (the coverage probability without noise does not depend on the density of the BSs). 

     We first introduce some notation about integer partitions from
number theory
that we use to present the main results in this paper. We need integer
partition to represent coefficient of $m$-{th} term of a polynomial which is a
product of a number of polynomials. The integer partition of
positive integer $k$ is a  way of writing $k$ as a sum of positive
integers. The set of all integer partitions of $k$ is denoted by $\mathcal{I}_k$
and $p_j$ is the $j$-th term in the partition $p$. Here we used $|p|$ to
denote the cardinality of the set $p$. For example, the integer partitions of 4
are given by $\mathcal{I}_4=\{ \{1,1,1,1\},\{1,1,2\},\{2,2\},\{1,3\},\{4\} \}$.
The second term of the partition, $p=\{1,1,1,1\}$, is   $p_2= 1$ and  $|p|=4$.
{For each partition, we introduce non-repeating partition set,
$q(p)$,
without any repeated summands  and $q(p)_j$ represents the number of times the
$j$-th  term of $q(p)$ is repeating in $p$. For example, for the partition
$p=\{1,1,2\}$, we have $q(p)=\{1,2\}$ and $q(p)_1=2$. The next Theorem provides
the coverage probability in a general setting.}
\begin{theorem}
The probability of coverage  with a linear MMSE receiver, when the BS locations are modelled by a PPP is 
\begin{align}
\pmmse(z)=\frac{1}{(1+z)^{N_t-1}} 
&\int_{0}^\infty e^{-z\sigma^2\,r^{\alpha}}\sum_{m=0}^{N_r-1}\sum_{v=1}^{N_r-m}
\frac{(z\sigma^2\,r^{\alpha})^{v-1}}{(v-1)!}z^m\sum_{k=0}^{\min(m,N_t-1)} {N_t
- 1 \choose k}\, \nonumber\\
&\sum_{p\in \mathcal{I}_{m-k}} e^{-\pi  \lambda  r^2
\Theta_{0,N_t}(z)} \left( 2\,
\pi\,  \lambda\right)^{|p|+1} \mathcal{K}_p(z) r^{2|p|+1} \d r,
\label{eq:covprobmain}
\end{align}
where $\Theta_{\varsigma, N_t}(z)=\,_2F_1\left(N_t,\varsigma-\frac{2}{\alpha
};\varsigma-\frac{2}{\alpha   }+1;-z \right)$ and 
$\mathcal{K}_p(z)=\frac{\prod_{j=1}^{|p|}{N_t \choose p_j}\frac{\Theta_{p_j,N_t}(z)}{\alpha  p_j-2}}{\prod_{j=1}^{|q(p)|}q(p)_j!}$.
\label{the:covmmse}
\end{theorem}
  \begin{proof}
 See Appendix C.
\end{proof}
When $\sigma^2=0$,  the coverage probability expressions can be simplified and does not require integration.
\begin{lemma} The coverage probability of a typical user  in an interference-limited  environment, \ie, $\sigma^2=0$  with MMSE receiver is 
\begin{eqnarray}
\pmmse(z)&=&\frac{1}{(1+z)^{N_t-1}}\sum_{m=0}^{N_r-1}
z^m \sum_{k=0}^{\min(m,N_t-1)} {N_t - 1 \choose k}\sum_{p\in \mathcal{I}_{m-k}}\frac{2^{|p|} \Gamma (|p|+1) \mathcal{K}_p(z)}{\Theta_{0, N_t}(z)^{|p|+1}}
\label{eq:prob_sc},
\end{eqnarray}
\end{lemma}
\begin{proof}
Follows from Theorem \eqref{the:covmmse}, by setting $\sigma^2=0$ and integrating with respect to $r$.
\end{proof}
Note the   the coverage expression in \eqref{eq:prob_sc} is not a function of $\lambda$. This is because of the  scale invariance property of the PPP.
In Figure \ref{fig:covfixN_r}, the coverage probability with linear MMSE
receiver is plotted  for different configurations. As expected, using higher
number of transmitting antennas, keeping $N_r$ the same, the coverage
probability reduces because of the increased interference. Also, as the number
of receiving antennas increase, keeping $N_t$ a constant, the coverage
probability increases  because of increased diversity order.
%
%

\begin{figure}[h]	
\centering
\begin{tikzpicture}
\begin{semilogyaxis}[scale=1.8,
xmin=-2,
xmax=16,
ymin=0.05,
grid = both,
legend style={
 cells={anchor=west},
legend pos=south west ,
 },
xlabel =Threshold $z$ dB,
ylabel = Coverage probability 
]
   \addplot[ mark=triangle*, mark size=3,  mark options={fill=white} ]
coordinates{( -5 , 0.94998 ) ( -4 , 0.93132 ) ( -3 , 0.90778 ) ( -2 , 0.87905 ) ( -1 , 0.84516 ) ( 0 , 0.80649 ) ( 1 , 0.76377 ) ( 2 , 0.71798 ) ( 3 , 0.67027 ) ( 4 , 0.62177 ) ( 5 , 0.57351 ) ( 6 , 0.52638 ) ( 7 , 0.48107 ) ( 8 , 0.43807 ) ( 9 , 0.39768 ) ( 10 , 0.36008 ) ( 11 , 0.32532 ) ( 12 , 0.29338 ) ( 13 , 0.26417 ) ( 14 , 0.23755 ) ( 15 , 0.21337 ) ( 16 , 0.19146 ) ( 17 , 0.17167 ) ( 18 , 0.15381 ) ( 19 , 0.13772 ) ( 20 , 0.12325 )};
 \addlegendentry{$1\times 2 $}
\addplot[ mark=square*, mark size=3,  mark options={fill=white} ]
coordinates{
 ( -5 , 0.81083 ) ( -4 , 0.76009 ) ( -3 , 0.70377 ) ( -2 , 0.64368 ) ( -1 , 0.58198 ) ( 0 , 0.52086 ) ( 1 , 0.46222 ) ( 2 , 0.40753 ) ( 3 , 0.3577 ) ( 4 , 0.31313 ) ( 5 , 0.27383 ) ( 6 , 0.2395 ) ( 7 , 0.20969 ) ( 8 , 0.18389 ) ( 9 , 0.16156 ) ( 10 , 0.14221 ) ( 11 , 0.12541 ) ( 12 , 0.11079 ) ( 13 , 0.098019 ) ( 14 , 0.086834 ) ( 15 , 0.07701 ) ( 16 , 0.068361 ) ( 17 , 0.060731 ) ( 18 , 0.053986 ) ( 19 , 0.048014 ) ( 20 , 0.042721 )   };
 \addlegendentry{$2\times 2$}
 
 \addplot[ mark=*, mark size=3,  mark options={fill=white} ]
coordinates{
 ( -5 , 0.9975 ) ( -4 , 0.99528 ) ( -3 , 0.9915 ) ( -2 , 0.98537 ) ( -1 , 0.97603 ) ( 0 , 0.96255 ) ( 1 , 0.94419 ) ( 2 , 0.92047 ) ( 3 , 0.89128 ) ( 4 , 0.85694 ) ( 5 , 0.81811 ) ( 6 , 0.77569 ) ( 7 , 0.73071 ) ( 8 , 0.68423 ) ( 9 , 0.63721 ) ( 10 , 0.5905 ) ( 11 , 0.54481 ) ( 12 , 0.50069 ) ( 13 , 0.45855 ) ( 14 , 0.41866 ) ( 15 , 0.38121 ) ( 16 , 0.34627 ) ( 17 , 0.31386 ) ( 18 , 0.28396 ) ( 19 , 0.25647 ) ( 20 , 0.2313 )  };
 \addlegendentry{$1\times 4$}
 
  \addplot[ mark=diamond*, mark size=3,  mark options={fill=white} ]
coordinates{
( -5 , 0.97249 ) ( -4 , 0.95535 ) ( -3 , 0.9312 ) ( -2 , 0.89927 ) ( -1 , 0.85954 ) ( 0 , 0.81286 ) ( 1 , 0.76075 ) ( 2 , 0.7052 ) ( 3 , 0.64822 ) ( 4 , 0.59166 ) ( 5 , 0.53699 ) ( 6 , 0.48522 ) ( 7 , 0.43699 ) ( 8 , 0.39259 ) ( 9 , 0.35208 ) ( 10 , 0.31536 ) ( 11 , 0.28221 ) ( 12 , 0.25239 ) ( 13 , 0.22562 ) ( 14 , 0.20161 ) ( 15 , 0.18011 ) ( 16 , 0.16087 ) ( 17 , 0.14366 ) ( 18 , 0.12827 ) ( 19 , 0.11451 ) ( 20 , 0.10222 )  };
 \addlegendentry{$2\times 4$}

  \addplot[ mark=*, mark size=3,  mark options={fill=black} ]
coordinates{
( -5 , 0.85008 ) ( -4 , 0.79192 ) ( -3 , 0.72487 ) ( -2 , 0.65202 ) ( -1 , 0.57708 ) ( 0 , 0.50362 ) ( 1 , 0.43459 ) ( 2 , 0.37196 ) ( 3 , 0.31674 ) ( 4 , 0.2691 ) ( 5 , 0.22865 ) ( 6 , 0.19465 ) ( 7 , 0.16624 ) ( 8 , 0.14255 ) ( 9 , 0.12277 ) ( 10 , 0.10621 ) ( 11 , 0.092273 ) ( 12 , 0.080484 ) ( 13 , 0.070451 ) ( 14 , 0.06186 ) ( 15 , 0.054462 ) ( 16 , 0.048056 ) ( 17 , 0.042484 ) ( 18 , 0.037616 ) ( 19 , 0.033349 ) ( 20 , 0.029596 )  };
 \addlegendentry{$4\times 4$}

 
%
%

\addplot[ mark=square*, mark size=3,  mark options={fill=black} ]
coordinates{
 (-4.999997 , 0.969822)  (-4.000002 , 0.957056)  (-3.000002 , 0.940199)  (-2.000002 , 0.918666)  (-1.000001 , 0.892101)  (0.000000 , 0.860471)  (1.000016 , 0.824096)  (1.999991 , 0.783631)  (2.999995 , 0.739973)  (4.000006 , 0.694151)  (5.000003 , 0.647215)  (5.999998 , 0.600137)  (6.999998 , 0.553756)  (7.999998 , 0.508750)  (8.999999 , 0.465632)  (10.000000 , 0.424759)  (11.000016 , 0.386360)  (11.999991 , 0.350554)  (12.999995 , 0.317375)  (14.000006 , 0.286793)  (15.000003 , 0.258731)  (15.999998 , 0.233078)  (16.999998 , 0.209705)  (17.999998 , 0.188468)  (18.999999 , 0.169218)  (20.000000 , 0.151806)  };
 \addlegendentry{$1\times 4$, PZF $m=2$}


\end{semilogyaxis}
\end{tikzpicture}
 		\caption{Coverage probability versus $z$  for for different antenna configurations with  $\sigma^2=0$ and $\alpha=4$ using a linear MMSE receiver. }
		\label{fig:covfixN_r}
\end{figure}
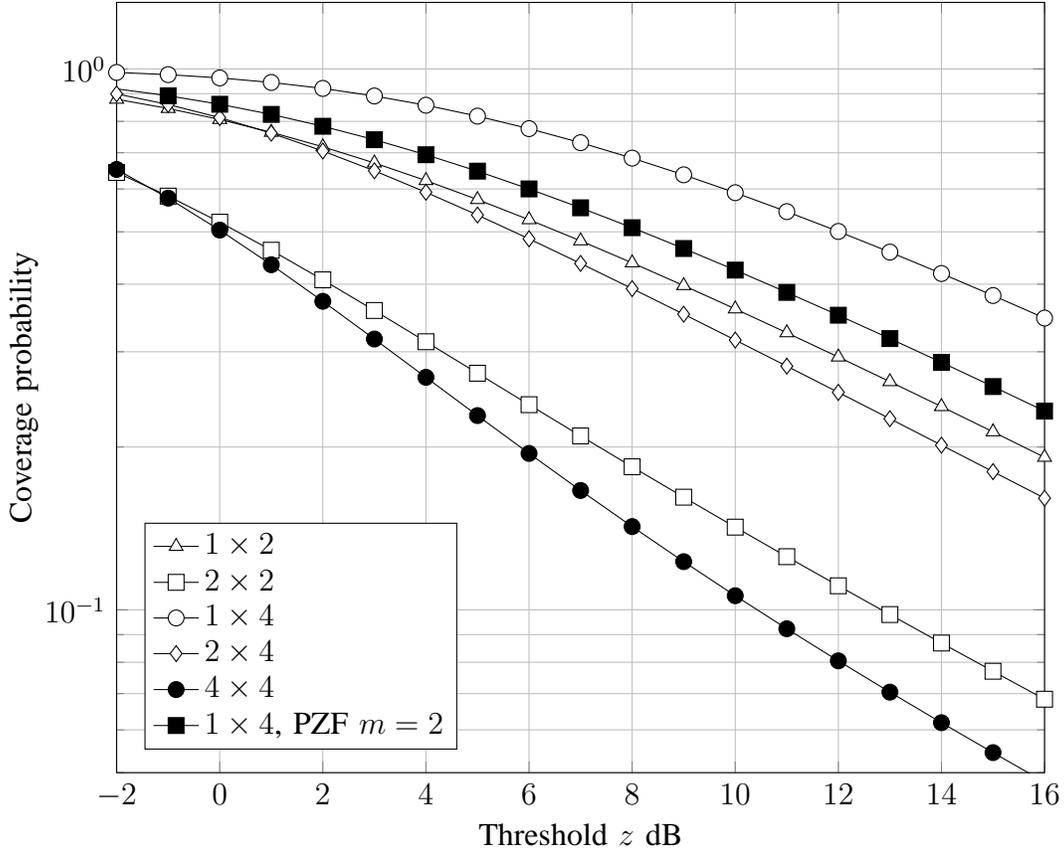
 
\section{Average ergodic rate}
\label{sec:rate}
In this Section, we compute the rate CDF for a typical user and also the ergodic data rate. We assume that $N_t$ users are being served by the BS in a cell, with one stream per user. Also for computing the rate, we treat residual interference as noise. The ergodic rate is given by $\E[\log_2(1+\sinr)]$ . Since $\log_2(1+\sinr)$ is a positive random variable, its mean is  given by the integral of its CCDF. Hence 
\begin{eqnarray}\label{eqn:rate}
\mathcal{C} (N_t,N_r)&\overset{\triangle}{=}&\E\left[\log_2(1+\sinr) \right] =\int_0^\infty \p(\sinr \geq 2^t-1)\d t.
 \end{eqnarray}
$ \p(\sinr \geq 2^t-1)$ depends on the receiver used and follows from the coverage probability by setting  $z= 2^t-1$. 

\noindent {\em Total rate with SM:} In SM each user decodes a single stream and hence achieves an ergodic rate   $\mathcal{C}(N_t,N_r)$, $N_t>1$. {\em  Hence for $N_t$ users}, the rate CDF is given by 
\begin{align}
F_{SM}(c)&= \P(N_t\log_2(1+\sinr)\leq c),  
\end{align}
where $\sinr $ denotes the  $\sinr$ with $N_t$ transmit and $m N_t+\delta$ receive antenna\footnote{In this expression we are neglecting the correlations of $\sinr$ across the users. }. The above distribution  can be easily computed from the $\sinr$ CCDF in Theorem \ref{the:cov} for a PZF receiver and the result in Theorem \ref{the:covmmse} for the MMSE receiver.  It is easy to see that the total average downlink rate  is given by $\mathcal{C}_{SM} = N_t\mathcal{C}(N_t,N_r)$.

\noindent{\em   Total rate with SST:} 
 In SST, the BS has only one antenna, \ie, $N_t =1$.  Hence it can serve only one stream and hence one user. So all the users are served  by dividing the resources either in time (TDMA) or frequency (FDMA). Hence in this case, each user has $1/N_t$ time or frequency slice. 
 In SST, since the resources have to be divided among the users, each user achieves an average rate
$N_t^{-1}\mathcal{C}(1,N_r)$. {\em  Hence for $N_t$ users}   the average total  downlink rate achieved is 
$\mathcal{C}_{SST} =\mathcal{C}(1,N_r)$.
The rate CDF is given by 
\[F_{SST}(c)= \P( \log_2(1+\sinr)\leq c).\]
 
\begin{table}
\begin{tabular}{cc}
\begin{tabular}{|c|c|c|c|c|c|} 
\hline                     
\multicolumn{6}{|c|}{PZF}\\\hline
SM/SST&$N_t\times N_r$ & $(m,\delta)$ & Mean & 5\% & 80\%\\ [0.5ex] 
\hline                  
\multirow{8}{*}{SST} &$1\times 4$&(4,0) & 3.51731 & 0.187 & 5.7162\\
&$1\times 4$&(3,1) & 4.2137 &   0.5804 &    6.4790 \\
&$1\times 4$&(2,2) & \underline{4.26918} & 0.7899   & 6.3964\\
&$1\times 4$&(1,3) & 3.83127 & 0.8333  & 5.622\\
&$1\times6$ & (6,0) & 3.9738 & 0.2585 &   6.3549\\
&$1\times6$ & (1,5) & 4.3858 & 1.2300 &    6.2300\\
&$1\times6$ & (2,4) & 5.0142 & 1.3100 &    7.2100\\
&$1\times6$ & (3,3) &\underline{ 5.2393} & 1.2900 &   7.5700\\\hline
\multirow{4}{*}{SM} &$2\times6$ & (3,0) & 5.0627 & 0.1541 &   8.4645\\
&$2\times6$ & (1,4) & 6.5011 & 1.2600 &   9.6600\\
&$2\times6$ & (2,2) &\underline{ 6.8904} & 0.9200 &  10.7\\
&$3\times6$ & (2,0) & 5.5691 & 0.1100 &  9.1100\\
&$3\times6$ & (1,3) & 7.6797 & 1.0500 &   11.7100\\      
\hline 
\end{tabular}&
\begin{tabular}{|c|c| c|c| c|} 
\hline                     
\multicolumn{5}{|c|}{MMSE}\\\hline
 SM/SST &$N_t\times N_r$ & Mean & 5\% &  80\%\\
\hline                  
\multirow{3}{*}{SST} &$1\times 2$ &   3.36 & 0.404  & 5.22 \\ 
&$1\times 4$ &  4.87  & 1.149 & 7.19 \\
& $1\times 6$ &  5.85  & 1.77 &  8.26 \\\hline
 \multirow{8}{*}{SM} &$2\times 2$ & 3.64 & 0.319 & 5.38 \\  
& $2\times 4$ &  6.35 & 1.016 &  9.46 \\
& $3\times 4$ &   6.63 & 0.958 & 9.73 \\
&$4\times 4$ &   6.58 & 0.925 &  9.09\\
&  $2\times 6$ &    8.06 & 1.73 & 11.97 \\
& $3\times 6$ &  9.12 & 1.66 & 13.64 \\
& $4\times 6$ &   9.56 & 1.62 & 13.96\\
&$5\times 6$ &    9.67 & 1.59 &  13.65\\
&$6\times 6$ &   9.34 & 1.57 &  12.87\\
 \hline 
\end{tabular}
\end{tabular}
 \caption{ Rate profile comparison for various configurations with  PZF (left) and MMSE (right) receivers. }
\label{table:rateprof_22}
\end{table}


 Various rate profiles are presented in in Table \ref{table:rateprof_22},  Figures \ref{figmmse},  \ref{figzfe}, \ref{fig:celledge} and \ref{fig:cell_edge_MMSE} for path loss exponent $\alpha =4$ and $\sigma^2=0$ obtained by numerically evaluating the analytical expressions.  
In Table \ref{table:rateprof_22}, the  rate profile is provided for various antenna configurations when a PZF receiver is used. We observe that the average rate is maximized\footnote{These maximum values are underlined in the Table.} when $m = m^*(0)=N_t\left\lceil\left(1-\frac{2}{\alpha}\right) \left( \frac{N_r}{N_t} -\frac{1}{2}\right)\right\rceil$.  We also observe that the MMSE receiver provides higher ergodic rate compared to PZF receiver for all antenna configurations. 
 
In Figure \ref{figmmse}, the average rate is plotted as a function of number of transmit streams  for various $N_r$ with a MMSE receiver. We observe that the average sum rate does not increase linearly with the number of transmit antenna. Interestingly, transmitting $N_t=N_r$ streams, does not lead to the maximum rate.  For example,  with $N_r =6$ the maximum rate is achieved by transmitting five streams and not six streams. 
We also observe that while  transmitting more streams than $N_r$, would hurt the average rate, the  rate reduction is slow with increasing streams. For example, consider the case of $N_r=4$. We see that transmitting five streams decreases the sum rate from $7$ bits/sec/Hz to $5$ bits/sec/Hz . However, the average rate is more or less fixed even if the number of streams are increased above five. From Figure \ref{figzfe}, similar to the MMSE receiver, we observe  diminishing returns with increasing $N_t$ even for the PZF receiver.
The mean rate for PZF receiver
configured with $1 \times 4$ is $4.27$ while it is $5.27$ for $2\times 4$, $5.57$ for $3\times 4$ and $4.47$ for $4 \times 4$. 
\begin{figure}	
\centering
\begin{tikzpicture}
\begin{axis}[scale=1.1,
grid = both,
legend style={
 cells={anchor=west},
legend pos=outer north east ,
 },
xlabel =\# of transmit antenna: $N_t$ (or \# of streams),
ylabel = Average sum rate 
]
   \addplot[ mark=diamond*, mark size=3,  mark options={fill=white} ]
coordinates{
( 1,2.14816)( 2,1.29352) ( 3,1.13337) ( 4,1.06637) ( 5,1.02962) ( 6,1.00641)( 7,0.990433)( 8,0.97876)( 9,0.96986)( 10,0.96285)(11,0.957186)( 12,0.952514)( 13,0.948595)( 14,0.94526)};
\addlegendentry{$N_r=1$}

\addplot[ mark=square*, mark size=3,  mark options={fill=white} ]
coordinates{
( 1,3.35546)( 2,3.64182)( 3,2.60299)( 4,2.32932)( 5,2.19738)( 6,2.11909)( 7,2.06713)( 8,2.03007)( 9,2.0023)( 10,1.98071)( 11,1.96344)( 12,1.94931)( 13,1.93753)( 14,1.92757)};
 \addlegendentry{$N_r=2$}

\addplot[ mark=triangle*, mark size=3,  mark options={fill=white} ]
coordinates{
( 1,4.86588) ( 2,6.24126)( 3,6.6364)( 4,6.5157)( 5,5.25587)( 6,4.81782)( 7,4.57299)( 8,4.41342)( 9,4.30034)( 10,4.21574)( 11,4.14995)
( 12,4.09727)( 13,4.05411)( 14,4.01809)};
 \addlegendentry{$N_r=4$}
 
 \addplot[ mark=*, mark size=3,  mark options={fill=white} ]
coordinates{
( 1,6.25335)( 2,8.78935)( 3,10.1473)( 4,10.8155)( 5,11.1406)( 6,11.156)( 7,10.7424)( 8,9.27908)( 9,8.67146)( 10,8.29331)( 11,8.02743)
( 12,7.82795)( 13,7.67187)( 14,7.54602)};
 \addlegendentry{$N_r=7$}
 
  \addplot[ mark=diamond*, mark size=3,  mark options={fill=black} ]
coordinates{
( 1,7.19207)( 2,10.5663)( 3,12.6643)( 4,14.0204)( 5,14.8742)( 6,15.354)( 7,15.6477)( 8,15.7202)( 9,15.5212)( 10,14.9309)( 11,13.3284)
( 12,12.5978)( 13,12.1155)( 14,11.761)};
 \addlegendentry{$N_r=10$}
 
   \addplot[ mark=square*, mark size=3,  mark options={fill=black} ]
coordinates{
( 1,8.10448) ( 2,12.3237)( 3,15.193)( 4,17.255)( 5,18.7606)( 6,19.8622)( 7,20.628)( 8,21.1259)( 9,21.4942)( 10,21.7144)( 11,21.7678)
( 12,21.6268)( 13,21.2407)( 14,20.488)};
 \addlegendentry{$N_r=14$}
 
    \addplot[ style=dashed]
coordinates{
( 1,2.14816)( 2,3.64182)( 3,5.08872)( 4,6.5157)( 5,7.93131)( 6,9.3395)( 7,10.7424)( 8,12.1414)( 9,13.5374)( 10,14.9309)( 11,16.3224)( 12,17.7123)( 13,19.1008)( 14,20.488)};
 \addlegendentry{$N_r=N_t$}
\end{axis}
\end{tikzpicture}
\caption{Average rate versus number of transmit streams with a MMSE  receiver.}
\label{figmmse}
\end{figure}
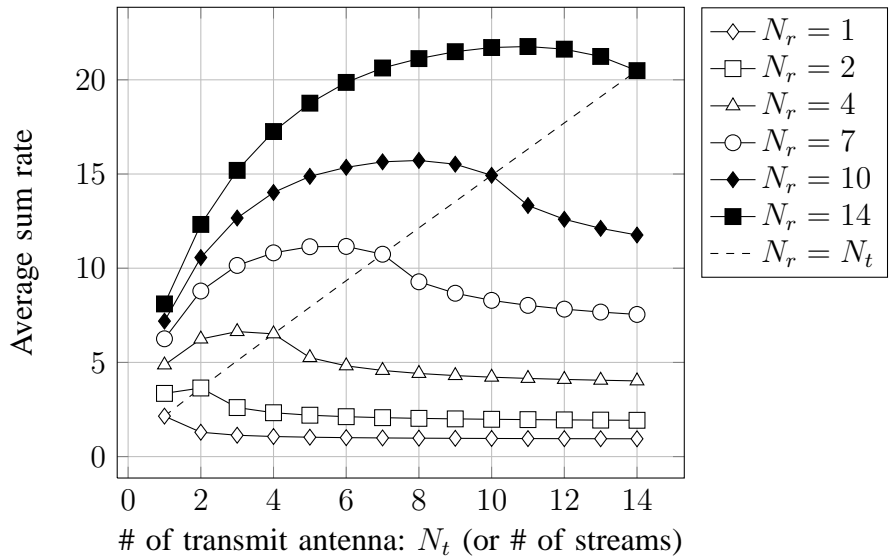
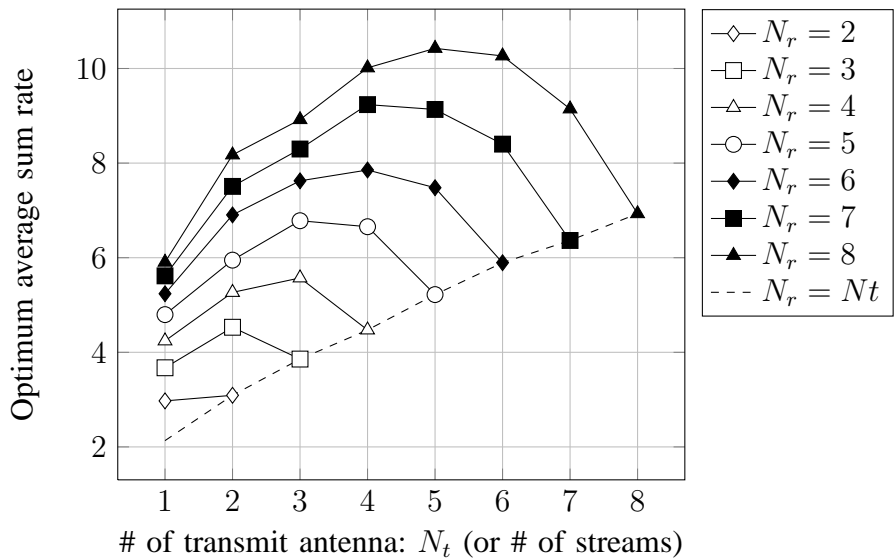
\begin{figure}	
\centering
\begin{tikzpicture}
\begin{axis}[scale=1.1,
grid = both,
legend style={
 cells={anchor=west},
legend pos=outer north east ,
 },
xlabel =\# of transmit antenna: $N_t$ (or \# of streams),
ylabel = Optimum average sum rate 
]

\addplot[ mark=diamond*, mark size=3,  mark options={fill=white} ]
coordinates{
(1.000000, 2.972727) (2.000000, 3.091040) };
\addlegendentry{$N_r = 2 $}
\addplot[ mark=square*, mark size=3,  mark options={fill=white} ]
coordinates{
(1.000000, 3.673222) (2.000000, 4.529492) (3.000000, 3.857072) };
\addlegendentry{$N_r = 3 $}
\addplot[ mark=triangle*, mark size=3,  mark options={fill=white} ]
coordinates{
(1.000000, 4.241344) (2.000000, 5.266351) (3.000000, 5.572716) (4.000000,
4.472605) };
\addlegendentry{$N_r = 4 $}
\addplot[ mark=*, mark size=3,  mark options={fill=white} ]
coordinates{
(1.000000, 4.796693) (2.000000, 5.948632) (3.000000, 6.781612) (4.000000,
6.654999) (5.000000, 5.220087) };
\addlegendentry{$N_r = 5 $}
\addplot[ mark=diamond*, mark size=3,  mark options={fill=black} ]
coordinates{
(1.000000, 5.238531) (2.000000, 6.904344) (3.000000, 7.625096) (4.000000,
7.854229) (5.000000, 7.479275) (6.000000, 5.893939) };
\addlegendentry{$N_r = 6 $}
\addplot[ mark=square*, mark size=3,  mark options={fill=black} ]
coordinates{
(1.000000, 5.616076) (2.000000, 7.507445) (3.000000, 8.297484) (4.000000,
9.236680) (5.000000, 9.134178) (6.000000, 8.401101) (7.000000, 6.364022) };
\addlegendentry{$N_r = 7 $}
\addplot[ mark=triangle*, mark size=3,  mark options={fill=black} ]
coordinates{
(1.000000, 5.905115) (2.000000, 8.173063) (3.000000, 8.919452) (4.000000,
10.013773) (5.000000, 10.425814) (6.000000, 10.265711) (7.000000, 9.145503)
(8.000000, 6.930855) };
\addlegendentry{$N_r = 8 $}
\addplot[ style=dashed ]
coordinates{
(1.000000, 2.133096) (2.000000, 3.091040) (3.000000, 3.857072) (4.000000,
4.472605) (5.000000, 5.220087) (6.000000, 5.893939) (7.000000, 6.364022)
(8.000000, 6.930855) };
\addlegendentry{$N_r = Nt $}
\end{axis}
\end{tikzpicture}
 		\caption{Average rate versus number of transmit streams with a
PZF receiver for optimal $m$. 
}
		\label{figzfe}
\end{figure}

  According to ITU definition, the $5\%$ point of the CDF of the normalized user throughput is considered as cell edge user
spectral efficiency and is plotted in Figure \ref{fig:celledge} for a PZF receiver. We  see
that increasing $N_t$ and hence increasing the number of streams in SM degrades
the performance of the edge users. For the  edge users the $\sinr$ is very weak.
Adding more streams will increase the interference which is difficult to
cancel. For example in the $N_r=4$ case, the mean rate increases  from $4.27$ to
$5.57$ when $N_t$ increases from $1$ to $3$. However,  cell edge users rate reduces 
 from $0.82$ to $0.34$ (almost halved) and for $N_t=4$
it is $0.073$. Therefore the degradation in performance for the cell edge users is
drastic compared to a little improvement in the average sum rate   for the
 PZF receiver when moving from SST to SM.  A similar observation can be made for other
 receiver configurations. 
 This implies increasing $N_t$ and using the multiple transmit antenna for
transmitting more streams will hurt the cell edge users. So from an edge user
perspective, SST is  more beneficial. 

\begin{figure}
\centering
\begin{tikzpicture}
\begin{axis}[scale=1.1,
grid = both,
legend style={
 cells={anchor=west},
legend pos=outer north east ,
 },
xlabel =\# of transmit antenna: $N_t$ (or \# of streams),
ylabel = Optimal cell edge spectral efficiency
]
\addplot[ mark=triangle*, mark size=3,  mark options={fill=white} ]
coordinates{
(1.000000, 0.329349) (2.000000, 0.074962) };
\addlegendentry{$N_r = 2 $}
\addplot[ mark=square*, mark size=3,  mark options={fill=white} ]
coordinates{
(1.000000, 0.594661) (2.000000, 0.351471) (3.000000, 0.078477) };
\addlegendentry{$N_r = 3 $}
\addplot[ mark=diamond*, mark size=3,  mark options={fill=white} ]
coordinates{
(1.000000, 0.829913) (2.000000, 0.686549) (3.000000, 0.349778) (4.000000,
0.073223) };
\addlegendentry{$N_r = 4 $}
\addplot[ mark=*, mark size=3,  mark options={fill=white} ]
coordinates{
(1.000000, 1.057610) (2.000000, 0.987761) (3.000000, 0.710276) (4.000000,
0.380057) (5.000000, 0.086100) };
\addlegendentry{$N_r = 5 $}
\addplot[ mark=triangle*, mark size=3,  mark options={fill=black} ]
coordinates{
(1.000000, 1.313221) (2.000000, 1.250251) (3.000000, 0.979055) (4.000000,
0.709232) (5.000000, 0.390032) (6.000000, 0.082081) };
\addlegendentry{$N_r = 6 $}
\addplot[ mark=diamond*, mark size=3,  mark options={fill=black} ]
coordinates{
(1.000000, 1.511698) (2.000000, 1.508651) (3.000000, 1.370245) (4.000000,
1.058666) (5.000000, 0.713417) (6.000000, 0.392296) (7.000000, 0.077835) };
\addlegendentry{$N_r = 7 $}
\addplot[ mark=square*, mark size=3,  mark options={fill=black} ]
coordinates{
(1.000000, 1.747619) (2.000000, 1.718281) (3.000000, 1.596350) (4.000000,
1.360914) (5.000000, 1.136129) (6.000000, 0.754784) (7.000000, 0.381360)
(8.000000, 0.079324) };
\addlegendentry{$N_r = 8 $}

\addplot[style=dashed]
coordinates{
(1.000000, 0.071346) (2.000000, 0.074962) (3.000000, 0.078477) (4.000000,
0.073223) (5.000000, 0.086100) (6.000000, 0.082081) (7.000000, 0.077835)
(8.000000, 0.079324) };
\addlegendentry{$N_r = Nt $}
\end{axis}
\end{tikzpicture}
 		\caption{Cell edge spectral efficiency versus number
of transmit streams with a PZF  receiver.}
		\label{fig:celledge}
\end{figure}
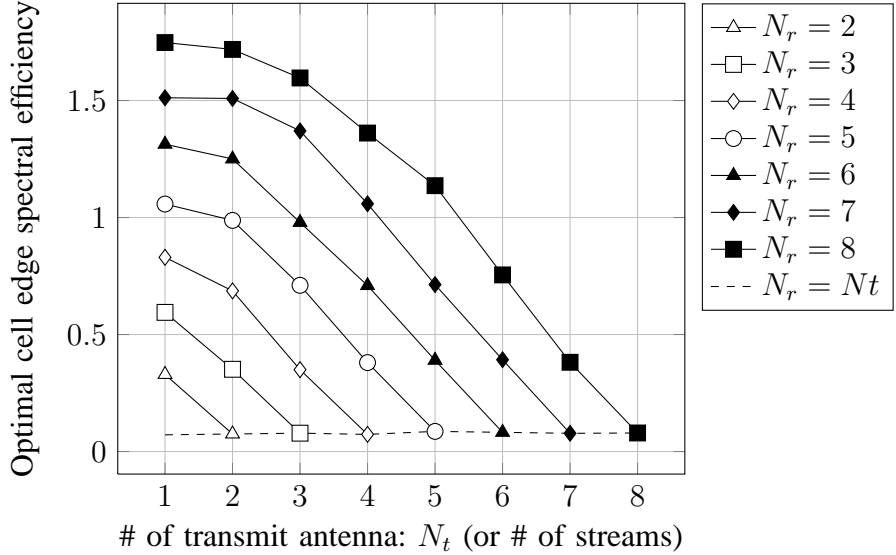

\begin{figure}
\centering
\begin{tikzpicture}
\begin{axis}[scale=1.1,
grid = both,
legend style={
 cells={anchor=west},
legend pos=outer north east ,
 },
xlabel =\# of transmit antenna: $N_t$ (or \# of streams),
ylabel = Cell edge  spectral efficiency
]
\addplot[ mark=diamond*, mark size=3,  mark options={fill=white} ]
coordinates{
(1.000000, 0.072049) (2.000000, 0.051020) (3.000000, 0.045627) (4.000000,
0.044555) (5.000000, 0.043318) (6.000000, 0.039798) (7.000000, 0.041100)
(8.000000, 0.043154) };
\addlegendentry{$N_r = 1 $}
\addplot[ mark=square*, mark size=3,  mark options={fill=white} ]
coordinates{
(1.000000, 0.383092) (2.000000, 0.310912) (3.000000, 0.271833) (4.000000,
0.276439) (5.000000, 0.265773) (6.000000, 0.252255) (7.000000, 0.254884)
(8.000000, 0.256947) };
\addlegendentry{$N_r = 2 $}
\addplot[ mark=triangle*, mark size=3,  mark options={fill=white} ]
coordinates{
(1.000000, 0.783185) (2.000000, 0.643557) (3.000000, 0.620894) (4.000000,
0.569229) (5.000000, 0.570496) (6.000000, 0.562161) (7.000000, 0.553880)
(8.000000, 0.545200) };
\addlegendentry{$N_r = 3 $}
\addplot[ mark=*, mark size=3,  mark options={fill=white} ]
coordinates{
(1.000000, 1.102469) (2.000000, 1.013859) (3.000000, 0.967499) (4.000000,
0.919787) (5.000000, 0.901516) (6.000000, 0.891403) (7.000000, 0.887057)
(8.000000, 0.878135) };
\addlegendentry{$N_r = 4 $}
\addplot[ mark=diamond*, mark size=3,  mark options={fill=black} ]
coordinates{
(1.000000, 1.462349) (2.000000, 1.397635) (3.000000, 1.350434) (4.000000,
1.277389) (5.000000, 1.237119) (6.000000, 1.201984) (7.000000, 1.209583)
(8.000000, 1.168319) };
\addlegendentry{$N_r = 5 $}
\addplot[ mark=square*, mark size=3,  mark options={fill=black} ]
coordinates{
(1.000000, 1.781815) (2.000000, 1.725574) (3.000000, 1.703482) (4.000000,
1.587151) (5.000000, 1.584179) (6.000000, 1.603093) (7.000000, 1.600523)
(8.000000, 1.532939) };
\addlegendentry{$N_r = 6 $}
\addplot[ mark=triangle*, mark size=3,  mark options={fill=black} ]
coordinates{
(1.000000, 2.090362) (2.000000, 2.088077) (3.000000, 1.976620) (4.000000,
1.993624) (5.000000, 1.893945) (6.000000, 1.981499) (7.000000, 1.911058)
(8.000000, 1.859809) };
\addlegendentry{$N_r = 7 $}
\addplot[ mark= *, mark size=3,  mark options={fill=black} ]
coordinates{
(1.000000, 2.341318) (2.000000, 2.339940) (3.000000, 2.376513) (4.000000,
2.350671) (5.000000, 2.291977) (6.000000, 2.218062) (7.000000, 2.242737)
(8.000000, 2.282743) };
\addlegendentry{$N_r = 8 $}

\end{axis}
\end{tikzpicture}
 		\caption{Cell edge spectral efficiency versus number
of transmit streams with a MMSE  receiver.}
		\label{fig:cell_edge_MMSE}
\end{figure}
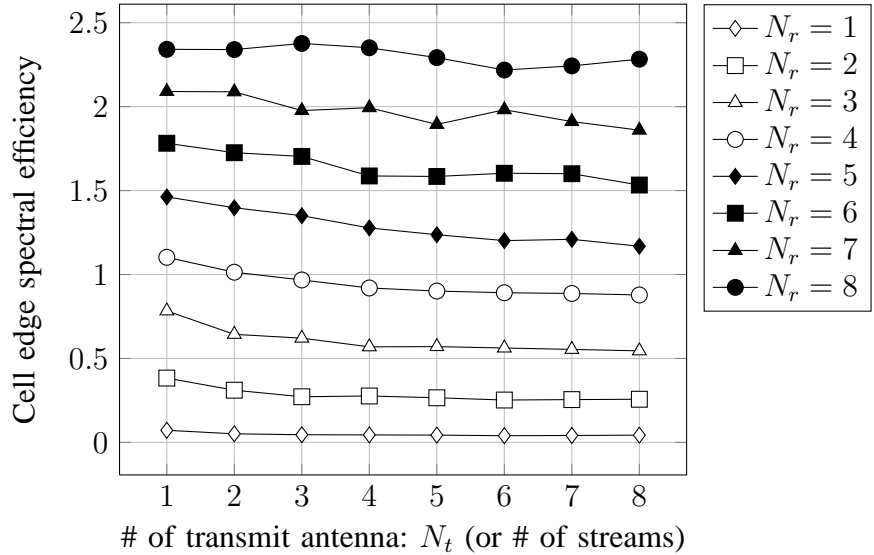
In Figure \ref{fig:cell_edge_MMSE}, the cell edge spectral efficiency is plotted for a MMSE receiver. We observe that unlike a PZF receiver, the cell edge rate does not decrease significantly with increasing streams.  This suggests that  MMSE receiver is the choice
for edge users if SM is utilized to transmit multiple streams. It will be
interesting to see the performance of MMSE receiver with limited channel
knowledge.

\section{Conclusion}
\label{conclusions}
In this paper, we characterized the performance of   open-loop  spatial
multiplexing   techniques in cellular networks with both MMSE and partial zero-forcing  receivers in the presence of distance dependent intercell interference.  Expressions for the CDF of the $\sinr$ of a typical user are obtained using tools from stochastic geometry. The distribution of $\sinr$ is used to characterize  the coverage and the rate of a typical user.  For the PZF receiver,  we show that it is optimal to cancel $N_t\left\lceil\left(1-\frac{2}{\alpha}\right) \left( \frac{N_r}{N_t} -\frac{1}{2}\right)\right\rceil$ closest interferers, where $\alpha$ is the path-loss exponent.  

We observe that increasing the SM rate provides an improvement in the mean rate with diminishing returns. The mean rate reaches a maximum value for a certain optimum SM rate that is generally less than $N_r$.  In contrast, increasing the SM rate always degrades the cell edge data rate for the PZF receiver, while the cell edge rate with a MMSE receiver is nearly independent of the SM rate. However, the MMSE receiver  requires the full channel knowledge and the practicality of  MMSE receiver should be addressed along with the pilot design methods  to enable reliable estimation of channel and interference parameters.
\section*{Acknowledgment}
 We would like to acknowledge the IU-ATC project for its
support. This project is funded by the Department of Science
and Technology (DST), India and Engineering and Physical
Sciences Research Council (EPSRC). We would also like to
thank the CPS project in IIT Hyderabad and the Samsung GRO project for supporting this
work.

\bibliographystyle{IEEEtran}
\bibliography{MMSE_PZF}

\appendices
\section{Proof of Theorem \ref{the:cov}}\label{sec:appA}
 In Theorem \ref{the:cov_gen}, setting $\sigma^2=0$, we have
\[\pzf(z)= \int_{0}^\infty \int_{r}^\infty \sum_{k=0}^{\delta}  \frac{(-s)^k}{k!} \frac{\d^k}{\d s^k} \L_{\I}(s)\Big|_{s=z r^\alpha}f_{R|r}(R|r) f_r(r) \d R   \d r.\] We first evaluate the  derivatives inside the above integral. 
The $k$-{th} derivative of $ \L_{\I}(s)$  can be evaluated using Fa\`{a} di
Bruno's formula for the 	 derivative of a composite function $g(f(x))$, and properties of the  derivatives of the 
hypergeometric function.  In this paper we  use a set partition version of
Fa\`{a} di Bruno's formula which is stated below.
\begin{align*} 
\frac{\d^k}{\d s^k} g\left(f(s)\right)=\sum_{\upsilon\in\mathcal{S}_k} g^{(|\upsilon|)}(f(s)) \prod_{j=1}^k (f^{(j)}(s))^{|\upsilon|_j},
\end{align*} 
where the notation for the  set partition is introduced in Section \ref{setpart}.
Let   $g(s)=e^{-\lambda\pi\,R^2 s}$ and $f(s)=\,  _2F_1\left(N_t,-\frac{2}{\alpha
   };\frac{\alpha -2}{\alpha };-R^{-\alpha } s\right)$.  Hence  $\L_{\I}(s)=e^{\lambda \pi   R^2}g(f(s))$.
The $p$-{th} derivative of $g(s)$ is  $ g^{(p)}(s)=\left(-\lambda \pi R^2\right)^p e^{-\lambda \pi R^2 s}
$. The following property of  hypergeometric functions can be easily verified $
\frac{\d^j}{\d s^j} \, _2F_1(a,b;c;s)=\frac{(a)_j (b)_j}{(c)_j}\, _2F_1(a+j,b+j;c+j;s)$, where $(a)_j $ is the Pochhammer symbol.  Hence, 
\begin{align} 
f^{(j)}(s)=\frac{(N_t)_j (-\frac{2}{\alpha
   })_j}{(\frac{\alpha -2}{\alpha })_j} (-R^{-\alpha })^j\, _2F_1(N_t+j,-\frac{2}{\alpha
   }+j;\frac{\alpha -2}{\alpha }+j;-R^{-\alpha }s) .
\end{align}
Using the property, $\sum_{j=1}^k j|\upsilon|_j=k$ and Fa\`{a} di Bruno's formula, we obtain 
\begin{align} 
\frac{\d^k}{\d s^k}\L_{\I_R}(s)\Big|_{s={zr^\alpha}}
   =&e^{\lambda \pi   R^2}\sum_{\upsilon\in\mathcal{S}_k} (-1)^{k+|\upsilon|}
\left(\lambda \pi   \right)^{|\upsilon|}\, R^{2|\upsilon|-\alpha k} 
e^{-\lambda\pi\,R^2\,  _2F_1\left(N_t,-\frac{2}{\alpha
   };\frac{\alpha -2}{\alpha };-R^{-\alpha } r^\alpha z \right) }\nonumber\\
  & \times\prod_{j=1}^k \left(\frac{(N_t)_j (-\frac{2}{\alpha
   })_j}{(\frac{\alpha -2}{\alpha })_j} \, _2F_1(N_t+j,-\frac{2}{\alpha
   }+j;\frac{\alpha -2}{\alpha }+j;-R^{-\alpha }r^\alpha z) \right)^{|\upsilon|_j}.
   \label{eq:ap1}
\end{align}

We obtain $\pzf(z)$ by substituting \eqref{eq:ap1} in Theorem \ref{the:cov_gen} with the functions
$f_{R|r}(R|r)$ and  $f_r(r)$ given by  \eqref{eqn:fRr} and
\eqref{eqn:fr}. Then by using the transformation $R/r \to \beta$ and $r\to t$ 
(which implies $\beta>1$),  and  the corresponding Jacobian we have (after basic algebraic manipulation), 
\begin{align*} 
\pzf(z)= \frac{4 (\pi \lambda)^m}{(m-2)!} \sum_{k=0}^{\delta} \frac{(-z
)^k}{k!}\int_{t=0}^\infty \int_{\beta=1}^\infty  \sum_{\upsilon\in\mathcal{S}_k}
(-1)^{k+|\upsilon|}
\left(\lambda \pi   \right)^{|\upsilon|}\, (\beta t)^{2|\upsilon|-\alpha k} 
e^{-\lambda\pi\,(\beta t)^2\,  _2F_1\left(N_t,-\frac{2}{\alpha
   };\frac{\alpha -2}{\alpha };-\beta^{-\alpha} z \right) }\nonumber\\
 \times\left[\prod_{j=1}^k \left(\frac{(N_t)_j (-\frac{2}{\alpha
   })_j}{(\frac{\alpha -2}{\alpha })_j} \, _2F_1(N_t+j,-\frac{2}{\alpha
   }+j;\frac{\alpha -2}{\alpha }+j;-\beta^{-\alpha } z)
\right)^{|\upsilon|_j}\right] t^{\alpha\,k+2m-1}(\beta^2  -1)^{m-2} \beta \d
\beta\d t.
\end{align*}
We can see that the product term is free of $t$ and we can group the other
terms and by integrating with respect to $t$ we obtain  the result.

\section{Proof of Theorem \ref{the:covmmse}}
Denote the  BS serving the typical  user at the origin by $x_0$ which is at a distance $r$,  we have from ~\cite{gao1998theoretical}, 
\begin{eqnarray}
\mathbb{P}(\sinr>z|r)&=&e^{-z/\gamma}\sum_{m=0}^{N_r-1}\left(\sum_{v=1}^{N_r-m} \frac{(z/\gamma)^{v-1}}{(v-1)!}\right) z^m\,\mathbb{E}_{\Phi^\prime}\left[\frac{C_m }{D(z)}\right]
\label{eq:condprob2},
\end{eqnarray} where  $N_r$ is the number of receiver antennas and  $C_m$ is the
coefficient of $z^m$ in $D(z)$, which is given by $
D(z)=(1+z)^{N_t-1}\prod_{x\in \Phi^\prime}(1+\Gamma_x z)^{N_t}$, where  $\Phi^\prime=\Phi^\setminus \{x_0\}$. The first term in $D(z)$ 
corresponds to same cell interference due to the $N_t-1$ streams intended for
the other users of the same cell and the second term  corresponds to the interference contribution  from other cells. Here, $\Gamma_x$'s are the interferer powers relative to the
desired source, \ie,  $\Gamma_x=\|x\|^{-\alpha}r^\alpha$.
By using the binomial expansion $D(z)$ can be expanded as
\begin{align*}
D(z)&=\sum_{k=0}^{N_t-1}{N_t - 1 \choose k}z^k \prod_{x\in \Phi^\prime}
\sum_{\nu=0}^{N_t}{N_t  \choose \nu} \Gamma_x^\nu z^\nu.
\end{align*}
We can observe that the coefficient of $z^m$ can be written
 as a product of coefficient of $z^k$ from the first polynomial and the
coefficient of $z^l$ from the second term  such
that $m=k+l$.  Hence the coefficient of $z^m$ is 
\begin{eqnarray}
C_m=\sum_{k+l=m} {N_t - 1 \choose k} \sum_{p\in
\mathcal{I}_l}\frac{\sum_{x_1,x_2,\hdots,x_{|p|}\in
\Phi^\prime}^{\neq }\prod_{j=1}^{|p|}{N_t \choose
p_j}  \Gamma_{x_i}^{p_j}}{\prod_{j=1}^{|q(p)|}q(p)_j!} ,
\label{eq:Cp}\end{eqnarray}
where $\mathcal{I}_l$ is the set of all integer partitions of $l$. See Section \ref{sec:mmse} for details   about integer partitions.  Here $\sum^{\neq}$ implies sums over disjoint tuples. 
The term $\prod_{j=1}^{|q(p)|}q(p)_j!$ in the denominator of \eqref{eq:Cp} is to
eliminate the repeating combinations of product terms formed by the
permutations of $x_1,x_2,\hdots,x_{|p|}\in\Phi^\prime$. For example,
$p=\{1,1,2\}$ is an integer partition of $l=4$,  and this partition will contribute product
terms $\{x_1x_2 x_3^2,\, x_2x_1 x_3^2,$ $\, x_1x_2^2 x_3,\, x_3x_2^2 x_1,\,
x_1^2x_2 x_3,\, x_1^2x_3 x_2\}$. Therefore the total
number of nonrepeating product terms is $3!/2!=3$. Hence
\begin{align}
\mathbb{E}_{\Phi^\prime}\left[\frac{C_m }{D(z)}\right] =(1+z)^{1-N_t}\sum_{k+l=m} \sum_{p\in
\mathcal{I}_l}\frac{ {N_t - 1 \choose k}}{\prod_{j=1}^{|q(p)|}q(p)_j! } \underbrace{\E\left[\frac{  \sum_{x_1,x_2,\hdots,x_{|p|}\in
\Phi^\prime}^{\neq }\prod_{j=1}^{|p|}{N_t \choose
p_j}  \Gamma_{x_i}^{p_j}}{\prod_{x\in \Phi^\prime}(1+\Gamma_x z)^{N_t}}\right]}_{T_1}. 
\label{eq:456}
\end{align}
We now focus on the term $T_1$, which can be rewritten as 
\[T_1= \E\left[  \sum_{x_1,x_2,\hdots,x_{|p|}\in
\Phi^\prime}^{\neq }\left(\prod_{j=1}^{|p|}{N_t \choose
p_j}  \frac{\Gamma_{x_i}^{p_j}}{(1+\Gamma_{x_i} z)^{N_t}}\right)  \prod_{x\in \Phi^\prime\setminus \{x_1,x_2,\hdots,x_{|p|}\}}(1+\Gamma_x z)^{-N_t}\right].\]
We now use Campbell-Mecke theorem for a PPP which we state for convenience.  Let $f(x, \phi): (\R^2)^n\times N \to [0,\infty)$ be a real valued function. Here $N$ denotes the set of all finite and simple sequences \cite{stoyan} in $\R^2$. Let $\Phi$ be a PPP of density $\lambda$.  We have
\[\E\sum_{x_1, x_2, \hdots x_n\in \Phi}^{\neq}f(x_1,x_2,\hdots x_n, \Phi \setminus \{x_1,x_2,\hdots, x_n\})= \lambda^n \int_{(\R^2)^n}\E[f(x_1,x_2,\hdots x_n, \Phi)]\d x_1\d x_2\hdots \d x_n. \]
In our case, we have  $T_1 = \E\sum_{x_1,x_2,\hdots,x_{|p|}\in
\Phi^\prime}^{\neq }f(x_1, x_2, \hdots, x_{|p|},\Phi^\prime\setminus \{x_1, x_2, \hdots, x_{|p|}\} )$, where 
\[f(x_1, x_2, \hdots, x_{|p|},\phi) = \underbrace{\left(\prod_{j=1}^{|p|}{N_t \choose
p_j}  \frac{\Gamma_{x_i}^{p_j}}{(1+\Gamma_{x_i} z)^{N_t}}\right) }_{T_2} \prod_{x\in \phi}(1+\Gamma_x z)^{-N_t}.\]
We use the probability generating functional of a PPP \cite{stoyan} to evaluate $\E_{\Phi^\prime}[f(x_1, x_2, \hdots, x_{|p|},\Phi^\prime)]$
\begin{align*}
\E_{\Phi^\prime}[f(x_1, x_2, \hdots, x_{|p|},\Phi^\prime)]&= T_2 \E \prod_{x\in \phi}(1+\Gamma_x z)^{-N_t}\\
&\stackrel{(a)}{=}T_2 \exp\left(-\lambda 2\pi \int_r^\infty x(1- (1+x^{-\alpha} r^\alpha z)^{-N_t})\d x\right)\\
&=T_2 \exp\left(-\pi  \lambda  r^2 \, \left(_2F_1\left(N_t,-\frac{2}{\alpha
};\frac{\alpha-2}{\alpha };-z\right)-1\right)\right).
\end{align*}
where $(a)$ follows from the PGFL of a PPP, polar coordinate transformation and the fact that the interferers are at a distance at least $r$ away. Now substituting in the Campbell-Mecke  theorem we obtain 
\begin{align*}
T_1= &\exp\left(-\pi  \lambda  r^2 \, \left(_2F_1\left(N_t,-\frac{2}{\alpha
};\frac{\alpha-2}{\alpha };-z\right)-1\right)\right)  \prod_{j=1}^{|p|}{N_t \choose p_j}2\pi \lambda\int_r^\infty   \frac{  x(x^{-\alpha}r^\alpha)^{p_j}}{(1+x^{-\alpha}r^\alpha z)^{N_t}}\d x, \nonumber \\
= & \exp\left({-\pi  \lambda  r^2 \left(\Theta_{0,N_t}(z)-1\right)}\right)  \prod_{j=1}^{|p|}{N_t \choose p_j}\frac{2 \pi  \lambda 
r^2_2F_1\left(N_t,p_j-\frac{2}{\alpha
   };p_j-\frac{2}{\alpha }+1;-z\right)}{\alpha  p_j-2},   \nonumber \\
   =&\exp\left({-\pi  \lambda  r^2
\left(\Theta_{0,N_t}(z)-1\right)}\right)  \left( 2\,
\pi\,  \lambda\,  r^2\right)^{|p|}\prod_{j=1}^{|p|}{N_t \choose p_j}\frac{\Theta_{p_j,N_t}(z)}{\alpha  p_j-2}.
 \nonumber 
\end{align*}
Substituting for $T_1$ in \eqref{eq:456} we have $\mathbb{E}_{\Phi^\prime}\left[\frac{C_m }{D(z)}\right]$ equals
\begin{align*}
(1+z)^{1-N_t}\exp\left({-\pi  \lambda  r^2
\left(\Theta_{0,N_t}(z)-1\right)}\right) \sum_{k+l=m} {N_t - 1 \choose k} \sum_{p\in
\mathcal{I}_l}\frac{ \left( 2\,
\pi\,  \lambda\,  r^2\right)^{|p|}\prod_{j=1}^{|p|}{N_t \choose p_j}\frac{\Theta_{p_j,N_t}(z)}{\alpha  p_j-2}}{\prod_{j=1}^{|q(p)|}q(p)_j! }.
\end{align*}
Substituting in  \eqref{eq:condprob2}, we obtain the conditional coverage probability. Averaging with respect to the density of $r$ given in \eqref{eqn:fr}, we obtain the result. 

\end{document}